\documentclass[reqno]{amsart}
\usepackage{amsmath,amsthm,amsrefs,verbatim}
\usepackage{amssymb}
\usepackage{graphicx}
\usepackage{bbm}
\usepackage{subfigure}
\usepackage[rgb]{xcolor}
\usepackage{mathrsfs}
\usepackage{comment}
\usepackage{stackrel} 
\usepackage{geometry}
\usepackage[all]{xy}
\usepackage{tikz}
\usepackage{wasysym} 
\usepackage{stmaryrd} 
\usepackage{cancel} 
\usepackage{float} 
\usepackage{numprint} 
\usepackage{enumitem} 
\usepackage{afterpage} 
\usepackage{arydshln} 

\usepackage{tikz}
\usetikzlibrary{matrix}

\newcommand{\beq}{\begin{equation}}
\newcommand{\eeq}{\end{equation}}

\newcommand{\C}{\mathbb{C}}
\newcommand{\N}{\mathbb{N}} 
\newcommand{\R}{\mathbb{R}} 
\renewcommand{\1}{\mathbbm{1}} 

\newcommand{\D}{\mathcal{D}_{p,q}}
\newcommand{\fR}{\mathfrak{R}}
\newcommand{\fI}{\mathfrak{I}}
\newcommand{\tr}{\mathrm{tr}}

\newcommand{\cS}{\mathcal{S}}
\newcommand{\fS}{\mathfrak{S}}
\newcommand{\cG}{\mathcal{G}}
\newcommand{\cH}{\mathcal{H}}
\newcommand{\cT}{\mathcal{T}}
\newcommand{\cU}{\mathcal{U}}
\newcommand{\Sp}{\mathrm{Sp}} 

\definecolor{myurlcolor}{rgb}{0,0,0.4}
\definecolor{mycitecolor}{rgb}{0,0.5,0}
\definecolor{myrefcolor}{rgb}{0.5,0,0}
\usepackage{hyperref}
\hypersetup{colorlinks,
linkcolor=myrefcolor,
citecolor=mycitecolor,
urlcolor=myurlcolor}

\numberwithin{equation}{section}
\makeatletter
\renewcommand*\env@matrix[1][*\c@MaxMatrixCols c]{%
  \hskip -\arraycolsep
  \let\@ifnextchar\new@ifnextchar
  \array{#1}}
\makeatother

\makeatletter
\newtheorem*{rep@theorem}{\rep@title}
\newcommand{\newreptheorem}[2]{%
\newenvironment{rep#1}[1]{%
 \def\rep@title{#2 \ref{##1}}%
 \begin{rep@theorem}}%
 {\end{rep@theorem}}}
\makeatother

\newtheorem{theorem}{Theorem} 
\newreptheorem{theorem}{Theorem}
\newtheorem{lemma}[theorem]{Lemma}
\newreptheorem{lemma}{Lemma}
\newtheorem{conj}[theorem]{Conjecture}
\newtheorem{corol}[theorem]{Corollary}
\newtheorem{defn}[theorem]{Definition}

\newtheorem{rem}[theorem]{Remark}

\usepackage{mathtools}
\usepackage[normalem]{ulem}

\newcommand{\ie}{\textit{i}.\textit{e}., }
\newcommand{\eg}{\textit{e}.\textit{g}. }

\def\d{\mathrm{d}}

\DeclareMathOperator\artanh{artanh}

\begin{document}

\title{$SU(p,q)$ coherent states and a Gaussian de Finetti theorem}
\author{Anthony Leverrier}
\address{Inria Paris, France.}
\email{anthony.leverrier@inria.fr}

\date{}

\begin{abstract}
We prove a generalization of the quantum de Finetti theorem when the local space is an infinite-dimensional Fock space. In particular, instead of considering the action of the permutation group on $n$ copies of that space, we consider the action of the unitary group $U(n)$ on the creation operators of the $n$ modes and define a natural generalization of the symmetric subspace as the space of states invariant under unitaries in $U(n)$. Our first result is a complete characterization of this subspace, which turns out to be spanned by a family of generalized coherent states related to the special unitary group $SU(p,q)$ of signature $(p,q)$. More precisely, this construction yields a unitary representation of the noncompact simple real Lie group $SU(p,q)$. We therefore find a dual unitary representation of the pair of groups $U(n)$ and $SU(p,q)$ on an $n(p+q)$-mode Fock space.  

The (Gaussian) $SU(p,q)$ coherent states resolve the identity on the symmetric subspace, which implies a Gaussian de Finetti theorem stating that tracing over a few modes of a unitary-invariant state yields a state close to a mixture of Gaussian states. As an application of this de Finetti theorem, we show that the $n\times n$ upper-left submatrix of an $n\times n$ Haar-invariant unitary matrix is close in total variation distance to a matrix of independent normal variables if $n^3 =O(m)$.
\end{abstract}

\thanks{I am extremely grateful to Matthias Christandl who was instrumental in the success of this project. I also gladly acknowledge inspiring discussions with Nicolas Cerf and Tobias Fritz and thank Robert K\"onig and Vivien Londe for comments on an early version of this manuscript.}

\maketitle

\section{Introduction}

Let $H_A$ and $H_B$ be two finite-dimensional Hilbert spaces of dimension $n$, and form the space $H = \bigoplus_{i=1}^p H_A \oplus \bigoplus_{j=1}^q H_B \cong \C^{n(p+q)}$, for integers $p,q \geq 1$. The Fock space $F(H)$ associated with $H$ is given by the direct sum of the symmetric tensors of $H$, $F(H) = \bigoplus_{k=0}^\infty \mathrm{Sym}^k (H)$. The Segal-Bargmann representation turns this infinite-dimensional space into a Hilbert space of holomorphic functions of $n(p+q)$ variables $z_{1,1}, \ldots, z_{n,p}, z'_{1,1}, \ldots, z'_{n,q}$ with a finite Gaussian measure. 
The unitary group $U(n)$ acts in a natural way on this space through a change of variables applied to the $(p+q)$ vectors of length $n$ given by $\vec{z_i} = (z_{1,i}, \ldots, z_{n,i})$ and $\vec{z_j}' = (z'_{1,j}, \ldots, z'_{n,j})$. For instance, any unitary $u \in U(n)$ acts on these vectors as
\begin{align}
\vec{z_i} \to u \vec{z_i} \quad \text{for} \, i \in [p] \quad \text{and} \quad \vec{z_j}' \to \overline{u} \vec{z_j}' \quad \text{for} \, j \in [q],  \label{eqn:action}
\end{align}
where $\overline{u}$ denotes the complex conjugate of $u$.

Physically, the Fock space $F_{p,q,n} = F(H)$ is particularly relevant in quantum optics since it describes an $n(p+q)$-dimensional harmonic oscillator and the action of the unitary group given by Eq.~\eqref{eqn:action} corresponds to that of passive linear optics networks consisting of beamsplitters and phase shifters in phase space.
Certain protocols in quantum communication with continuous variables, for instance in quantum cryptography \cite{SBC08}, are covariant with respect to this action and it is therefore relevant for their study to understand which states of $F_{p,q,n}$ are left invariant under the action of all unitaries in $U(n)$. We denote the corresponding subspace by $F_{p,q,n}^{U(n)}$.

The space $F_{p,q,n}^{U(n)}$ is a natural generalization of the usual \emph{symmetric subspace} $\mathrm{Sym}^n (\mathbbm{C}^d) \subset (\mathbbm{C}^d)^{\otimes n}$ of states invariant under all permutations of the $d$-dimensional subsystems, when the local space $\C^d$ is replaced by an infinite-dimensional Fock space, provided that the system is invariant under the action of $U(n)$ and not simply under the action of the symmetric group $S_n$. 
Examples of applications include state estimation and cloning of continuous-variable states (see Ref.~\cite{chi10} for such applications in the finite-dimensional case) as well as the security analysis of quantum cryptography protocols with continuous variables \cite{lev17}. 
The notion of invariance under the unitary group can be extended to density matrices, positive semidefinite matrices of trace 1, upon which the unitary group acts by conjugation. For instance, the \emph{twirling map} defined by
\begin{align*}
\begin{array}{cccc}
\mathcal{T}:& \mathfrak{S}(F_{p,q,n}) &\to &   \mathfrak{S}(F_{p,q,n})   \\
&\rho_{AB} & \mapsto & \int W_u \rho_{AB} W_u^\dagger \mathrm{d}u
\end{array}
\end{align*}
where $W_u$ is the representation of the unitary $u$ on $F_{p,q,n}$, $W_u^\dagger$ denotes the adjoint (conjugate transpose) of $W_u$, and the averaging is performed with respect to the Haar measure on the unitary group $U(n)$, maps arbitrary quantum states on $F_{p,q,n}$ to invariant states and is a natural (infinite-dimensional) generalization of the twirling map defined for $2$-qudit states on $\C^d \otimes \C^d$ that has found numerous applications in quantum information theory (see for instance \cite{wat16}).

The outline of the paper is as follows.

We first introduce the symmetric subspace $F_{p,q,n}^{U(n)}$ and provide a complete characterization of that subspace in Section \ref{sec:symm-sub}. In particular, we show that it coincides with the space of square-integrable holomorphic functions in the $pq$ variables $Z_{1,1}, \ldots, Z_{p,q}$ where the variable $Z_{i,j}$ is given by $\sum_{k=1}^n z_{k, i} z'_{k,j}$.
\begin{theorem}[Characterization of the symmetric subspace]\label{thm:charact-symm}
For $p, q \geq 1$ and $n \geq \min(p,q)$, the symmetric subspace $F_{p,q,n}^{U(n)}$ is isomorphic to the space of square-integrable holomorphic functions in the variables $Z_{1,1}, \ldots, Z_{p,q}$, with the norm induced by the one on $F_{p,q,n}$.
\end{theorem}
Similarly, a density operator $\rho$ acting on $F_{p,q,n}$ is said to be invariant under the action of the unitary group if $W_u \rho W_u^\dagger = \rho$ for all $u\in U(n)$. We prove that invariant density matrices always admit an invariant purification in $F_{p+q, p+q,n}^{U(n)}$. 
\begin{theorem} \label{theo:purification}
Any density operator $\rho \in \mathfrak{S}(F_{p,q,n})$ invariant under the action of the unitary group $U(n)$ admits a purification in $F_{p+q, p+q,n}^{U(n)}$.
\end{theorem}

We show in Section \ref{sec:CS}  that the symmetric subspace $F_{p,q,n}^{U(n)}$ carries an irreducible unitary representation of the generalized special unitary group $SU(p,q)$, which is noncompact for $p,q\geq 1$. While the case $p=q=1$ has been extensively studied in the literature, starting with Bargmann \cite{bar47}, Gelfand and Neumark \cite{GN46}, it seems that the general case has not received the same attention until now. The main exception is the work of Perelomov who defined coherent states $|\Lambda,n\rangle$ for the Lie group $SU(p,q)$ in Ref.~\cite{per72}, where $\Lambda$ belongs to the set $\D$ of $p\times q$ complex matrices with spectral norm strictly less than 1. The main property of these coherent states is that they resolve the identity on $F_{p,q,n}^{U(n)}$, when averaged with the invariant measure $\d \mu_{p,q,n}(\Lambda)$ defined in Eq.~\eqref{eqn:mu}.
\begin{theorem}[Resolution of the identity]
\label{thm:resol}
For $n \geq p+q$, the generalized $SU(p,q)$ coherent states resolve the identity on $F_{p,q,n}^{U(n)}$:
\begin{align*}
\int_{\D} |\Lambda,n\rangle \langle \Lambda,n| \mathrm{d}\mu_{p,q,n}(\Lambda) = \mathbbm{1}_{F_{p,q,n}^{U(n)}}.
\end{align*}
\end{theorem}
Since the $SU(p,q)$ coherent states are \emph{Gaussian} states, in the sense that their characteristic function is Gaussian \cite{WPG12}, they are entirely characterized by their first two moments. We analyze in Section \ref{sec:gauss} the phase-space properties of the unitary representation of $SU(p,q)$, notably how the covariance matrix of $|\Lambda, n\rangle$ depends on the matrix $\Lambda$ and how $SU(p,q)$ is mapped to the symplectic group $Sp(2(p+q), \R)$ of Gaussian operations in phase-space. 

We  prove a de Finetti theorem stating that tracing over a few modes of a unitary-invariant state gives a state close to a mixture of Gaussian states in Section \ref{sec:dF}. 
\begin{theorem}[Gaussian de Finetti] \label{thm:finetti}
Let $n$ be an arbitrary integer and $k\geq p+q$. Let $\rho = |\psi\rangle\langle \psi|$ be a symmetric (pure) state in $F_{p,q,n+k}^{U(n+k)}$. Then the state obtained after tracing out over $k(p+q)$ modes can be well approximated by a mixture of generalized coherent states: 
\begin{align*}
\left\|\tr_k  (\rho) - C_k \int \nu(\Lambda) |\Lambda,n \rangle \langle \Lambda, n| \d\mu_{p,q}(\Lambda) \right\|_{\tr}\leq \frac{3npq}{2(n+k-p-q)},
\end{align*}
with the density $\nu(\Lambda) := |\langle \Lambda, n+k| \psi\rangle|^2$ and $\d \mu_{p,q} := \frac{1}{C_n} \d \mu_{p,q,n}$.
\end{theorem}
This theorem is obtained as a special case of a general de Finetti theorem for representations of symmetry groups established by K\"onig and Mitchison \cite{KM09}. 

In Sections \ref{sec:11} and \ref{sec:22}, we study in more detail the important special cases $p=q=1$ and $p=q=2$: the former has been extensively studied in the literature under the name of $SU(1,1)$ coherent states and has found numerous applications in mathematical physics; the latter has been far less explored but turns out to be particularly relevant for studying for instance some natural continuous-variables quantum key distribution protocols. Indeed, such protocols involve mixed states invariant under the action of $U(n)$ on $F_{1,1,n}$ and these states admit a purification in the symmetric subspace $F_{2,2,n}^{U(n)}$ as proven in Theorem \ref{theo:purification}. While an explicit orthonormal basis for $F_{1,1,n}^{U(n)}$ is known, the task appears more complex in the case where $p=q=2$, and we are only able to conjecture such an explicit orthonormal basis for $F_{2,2,n}^{U(n)}$.

In Section \ref{sec:marginal}, we give an application of the Gaussian de Finetti theorem to the study of submatrices of random Haar unitary matrices. More precisely, let $\cH_{m,n}$ be the distribution over $n \times n$ complex matrices obtained by first drawing a unitary $U$ from the Haar measure on $U(m)$ and then outputting $\sqrt{m} U_{n,n}$ where $U_{n,n}$ is the $n \times n$ upper-left submatrix of $U$.
Let $\cG^{n\times n}$ be the probability distribution over $n \times n$ complex matrices whose entries are independent Gaussians with mean 0 and variance 1. Then we show the following:
\begin{theorem}\label{thm:trunc}
Let $m \geq n$. Then $\|\cH_{m,n} - \cG^{n\times n}\|_{\mathrm{TV}} \leq \frac{2n^3}{m-n}$, where $\|\cdot\|_{\mathrm{TV}}$ denotes the total variation distance. 
\end{theorem}
While it was recently proved by Jiang and Ma \cite{JM17} that the distance goes asymptotically to 0 for $n=o(m^{1/2})$, our result follows almost directly from a application of our de Finetti theorem and has the advantage of providing an explicit bound on the total variation distance. 

We conclude by mentioning that while the well-known Schur-Weyl duality between $GL(d, \C)$ and $S_n$ has received a wide attention in quantum information theory, it does not seem to be the case for the duality between $SU(p,q)$ and $U(n)$ that we explore in the present paper. Besides its direct application to some specific quantum cryptography protocols, we expect this duality to be relevant for the study of many tasks in quantum information with continuous variables involving many modes.


\section{The symmetric subspace $F_{p,q,n}^{U(n)}$}
\label{sec:symm-sub}
Let  $H_A$, $H_B$ be two Hilbert spaces of dimension $n$ and define the $n(p+q)$-dimensional Hilbert space $H_{p,q,n} := H_A^{\oplus p} \oplus H_B^{\oplus q} \cong \C^{n(p+q)}$ for integers $p, q$. 
The Fock space $F_{p,q,n}$ associated with $H_{p,q,n}$ is the infinite-dimensional Hilbert space  
\begin{align*}
F_{p,q,n} := \bigoplus_{k=0}^\infty \mathrm{Sym}^k( H_{p,q,n}),
\end{align*} 
where $\mathrm{Sym}^k (H)$ stands for the symmetric part of $H^{\otimes k}$.

Using the Segal-Bargmann representation, the Hilbert space $F_{p,q,n}$ is realized as a functional space of complex holomorphic functions square-integrable with respect to a Gaussian measure, $F_{p,q,n} \cong L^2_{\mathrm{hol}}(\C^{n(p+q)}, \| \cdot\|)$, where the state $\psi \in F_{p,q,n}$ is represented by a holomorphic function $\psi(z,z')$ with $z \in \C^{np}, z' \in \C^{nq}$ satisfying 
\begin{align} \label{eqn:norm}
\|\psi\|^2 := \langle \psi, \psi\rangle = \frac{1}{\pi^{n(p+q)}}\int \exp(-|z|^2 -|z'|^2) |\psi(z,z')|^2 \d z \d z'< \infty
\end{align}
where $\d z :=  \prod_{k=1}^n \prod_{i=1}^p \mathrm{d}z_{k,i}$ and $\d z' := \prod_{k=1}^n \prod_{j=1}^q \mathrm{d}z_{k,j}'$ denote the Lebesgue measures on $\C^{np}$ and $\C^{nq}$, respectively,  and $|z|^2 := \sum_{k=1}^n\sum_{i=1}^p |z_{k,i}|^2, |z'|^2 := \sum_{k=1}^n \sum_{j=1}^q |z_{k,j}'|^2$.
A state $\psi$ is therefore described as a holomorphic function of $n(p+q)$ complex variables $(z_{1,1}, \ldots z_{n,p},  z_{1,1}', \ldots z_{n,q}')$. In the following, we denote by $z_i$ and $z_j'$ the vectors $(z_{1,i}, \ldots, z_{n,i})$ and $(z_{1,j}', \ldots, z_{n,j}')$, respectively, for $i \in [p]$ and $j \in [q]$, where $[p] := \{1, 2, \ldots, p\}$.

The $n(p+q)$-mode Fock space factorizes as $F_{p,q,n} = F_A \otimes F_B$, where $F_A \cong F_{p,0,n}$ and $F_B \cong F_{0,q,n}$ are the Fock spaces associated with $H_A \cong \C^{np}$ and $H_B \cong \C^{nq}$, respectively.
Let $\mathfrak{B}(F_{p,q,n})$ denote the set of bounded linear operators from $F_{p,q,n}$ to itself and let $\mathfrak{S}(F_{p,q,n})$ be the set of quantum states on $F_{p,q,n}$: positive semi-definite operators with unit trace.

It will sometimes be convenient to use the physicist bra-ket notation and write $|\psi\rangle$ for the state $\psi$ and $\langle \psi |$ for its adjoint. Formally, one can switch from the Segal-Bargmann representation to the representation in terms of annihihation and creation operators by replacing the variables $z_{k,i}$ and $z'_{k,j}$ by the \emph{creation operators} $a_{k,i}^\dagger$ and $a_{k,j}'^\dagger$ of modes $(k,i)$ in $F_A$ and $(k,j)$ in $F_B$, respectively. 
The function $f(z,z')$ is therefore replaced by an operator $f(a^\dagger, a'^\dagger)$ and the corresponding state in the Fock basis is obtained by applying this operator to the vacuum state.
Similarly, the \emph{annihilation operators} $a_{k,i}$ and $a'_{k,j}$ on modes $(k,i)$ and $(k,j)$ correspond to partial derivatives for the corresponding variables: $\partial/\partial z_{k,i}$ and $\partial/\partial z'_{k,j}$. 
Recall that the annihilation and creation operators satisfy the canonical commutation relations:
\begin{align*}
&[a_{k_1,i_1},a_{k_2,i_2}] = [a_{k_1,i_1},a'_{k_2,j_2}] = [a'_{k_1,j_1},a'_{k_2,j_2}]  =[a^\dagger_{k_1,i_1},a^\dagger_{k_2,i_2}] = [a^\dagger_{k_1,i_1},a'^\dagger_{k_2,j_2}] = [a'^\dagger_{k_1,j_1},a'^\dagger_{k_2,j_2}]  =  0\\ 
&[a_{k_1,i_1},a^\dagger_{k_2,i_2}] = \delta_{k_1,k_2} \delta_{i_1,i_2}, \quad
[a_{k_1,i_1},a'^\dagger_{k_2,j_2}] = 0, \quad
[a'_{k_1,j_1},a'^\dagger_{k_2,j_2}]  = \delta_{k_1,k_2} \delta_{j_1,j_2}
\end{align*} 
where $\delta_{i,j}$ is the Kronecker symbol equal to 1 if $i=j$ and 0 otherwise.

The metaplectic representation of the unitary group $U(n) \subset Sp(2n,\R)$ on $F_A \cong F_{p,0,n}$ associates to $u \in U(n)$ the operator $V_u$ corresponding to the change of variables $z \to uz$, that is $V_u \psi(z) := \psi (u z)$. 
We will be interested in the following extension $W$ of this unitary representation on $F_{p,q,n}$.
\begin{defn}[Representation of the unitary group $U(n)$]
The unitary group $U(n)$ admits the following unitary representation on the Fock space $F_{p,q,n}$:
\begin{align*}
W: U(n) & \to \mathfrak{B}(F_{p,q,n})\\
u & \mapsto W_u := \big[ \psi(z_1, \ldots, z_p, z_1', \ldots, z_q') \mapsto \psi(u z_1, \ldots, u z_p, \overline{u} z_1', \ldots, \overline{u} z_q')\big]
\end{align*}
where $\overline{u}$ denotes the complex conjugate of the unitary matrix $u$.
\end{defn}
Note that the unitary $u$ is applied to the modes of $F_A$ and its complex conjugate is applied to those of $F_B$. 
 
The states that are left invariant under the action of the unitary group $U(n)$ are relevant for some applications of quantum information with continuous variables, and our main objective in this paper is to understand the subspace of $F_{p,q,n}$ spanned by such invariant states. 
\begin{defn}[Symmetric subspace]
For integers $p,q, n \geq 1$, the \emph{symmetric subspace} $F_{p,q,n}^{U(n)}$ is the subspace of functions $\psi \in F_{p,q,n}$ such that
\begin{align*}
W_u \psi = \psi \quad \forall u \in U(n).
\end{align*} 
\end{defn}
The name \emph{symmetric subspace} is inspired by the name given to the subspace $\mathrm{Sym}^n (\mathbbm{C}^d)$ of $(\mathbbm{C}^d)^{\otimes n}$ of states invariant under arbitrary permutations of the subsystems:
\begin{align}
\mathrm{Sym}^n (\mathbbm{C}^d) := \left\{|\psi\rangle \in(\mathbbm{C}^d)^{\otimes n} \: : \: P(\pi) |\psi\rangle = |\psi\rangle, \forall \pi \in S_n \right\}
\end{align}
where $\pi \mapsto P(\pi)$ is a representation of the permutation group $S_n$ on $(\mathbbm{C}^d)^{\otimes n}$ and $P(\pi)$ is the operator that permutes the $n$ factors of the state according to $\pi \in S_n$. See for instance \cite{har13} for a recent exposition of the symmetric subspace from a quantum information perspective.

For instance, in the case where $p=q=1$, it is straightforward to check that the action of the unitary group leaves the so-called \emph{two-mode squeezed vacuum state} with squeezing parameter $\lambda$,
$$\phi_\lambda (z,z') := (1-|\lambda|^2)^{n/2} \exp (\lambda (z_{1} z_{1}' + \ldots + z_n z_n')),$$ 
invariant, where $\phi_\lambda \in F_{1,1,n}$ for $\lambda \in \C$ such that $|\lambda|<1$.  
Indeed, the quadratic form $z_{1} z_{1}' + \ldots + z_n z_n'$ is invariant under the change of variable $z \to u z, z' \to \overline{u} z'$:
\begin{align} \label{eqn:Zinv}
\sum_{k=1}^n (u z)_{k} (\overline{u} z')_k &= \sum_{k=1}^n \sum_{i=1}^n \sum_{j=1}^n u_{k,i} z_i \overline{u}_{k,j} z'_j =   \sum_{i=1}^n \sum_{j=1}^n z_i  z'_j \sum_{k=1}^n u_{k,i} (u^\dagger)_{j,k} = \sum_{i=1}^n z_i z_j'
\end{align}
where the last equality results from the unitarity of $u$.
In fact, these two-mode squeezed vacuum states correspond to a realization of $SU(1,1)$ coherent states that we will discuss in more detail in Section \ref{sec:11}.

It will be useful to consider finite-dimensional subspaces of $F_{p,q,n}^{U(n)}$, where the holomorphic functions are restricted to polynomials of bounded total degree in the variables of type $z$, namely $z_{1,1}, \ldots, z_{n,p}$. In the context of quantum optics, such functions describe quantum states with a bounded total number of photons.
\begin{defn}
For integers $p,q ,n \geq 1$, and $d \geq 0$, the subspace $F_{p,q,n}^{U(n), \leq d}$ is given by
\begin{align*}
F_{p,q,n}^{U(n), \leq d} := \{P(z_{1,1} \ldots, z_{n,p}, z'_{1,1},\ldots, z'_{n,q}) \in F_{p,q,n}^{U(n)} \: : \: \mathrm{deg}_{(z_{1,1}, \ldots, z_{n,p})}(P) \leq d\}.
\end{align*}
\end{defn}

For $i \in [p], j \in [q]$, we introduce the variable $Z_{i,j}$ defined by
\begin{align}\label{eq:Zij}
Z_{i,j} := \sum_{k=1}^n z_{k,i} z_{k,j}'.
\end{align} 
A function of the variables $Z_{1,1}, \ldots, Z_{p,q}$ is naturally mapped to a function of the variables $z_{1,1} \ldots, z_{n,p}$, $z'_{1,1},\ldots, z'_{n,q}$ via the application 
\begin{displaymath}
\begin{array}{ccl}
\Phi: & \C[[Z_{1,1}, \ldots, Z_{p,q}]] & \to \C[[z_{1,1} \ldots, z_{n,p}, z'_{1,1},\ldots, z'_{n,q}]]\\
& f(Z_{1,1}, \ldots, Z_{p,q}) & \mapsto f\big(\sum_{k=1}^n z_{k,1} z_{k,1}', \ldots, \sum_{k=1}^n z_{k,p} z_{k,q}' \big) \\
\end{array}
\end{displaymath}
In particular, the norm $\|\cdot\|$ of Eq.~\eqref{eqn:norm} induces a norm $\|\cdot\|_E$ on $\C[[Z_{1,1}, \ldots, Z_{p,q}]]$ via the application $\Phi$: for a function $f \in \C[[Z_{1,1}, \ldots, Z_{p,q}]]$, 
\begin{align*}
\| f \|_{E} := \|\Phi \circ f\|.
\end{align*}
Note that $\Phi$ induces an isomorphism between $\C[[Z_{1,1}, \ldots, Z_{p,q}]]$ and its image. 

\begin{defn}
For integers $p,q, n\geq 1$, let $E_{p,q,n}$ be the space of analytic functions $\psi$ of the $pq$ variables $Z_{1,1}, \ldots, Z_{p,q}$, satisfying $\|\psi\|_E^2 < \infty$, that is $E_{p,q,n} = L^2_{\mathrm{hol}}(\C^{pq}, \|\cdot\|_E)$.
Moreover, for $d\in \N$, we define the finite-dimensional subspace $E_{p,q,n}^{\leq d}$ by
\begin{align*}
E_{p,q,n}^{\leq d} = \left\{ P(Z_{1,1}, \ldots, Z_{p,q}) \: : \: \mathrm{deg}(P) \leq d\right\}.
\end{align*}
\end{defn}
Note that $E_{p,q,n}$ depends on the integer $n$ via the norm $\|\cdot\|_E$. 
Our first result is a characterization of the symmetric subspace $F_{p,q,n}^{U(n)}$, which turns out to coincide with the space $E_{p,q,n}$ defined above.

\begin{reptheorem}{thm:charact-symm}
For $p, q \geq 1$ and $n \geq \min(p,q)$, the symmetric subspace $F_{p,q,n}^{U(n)}$ is isomorphic to $E_{p,q,n}$, and the isomorphism is given by $\Phi$.
\end{reptheorem}

One inclusion (up to the morphism $\Phi$) is straightforward to prove. 
\begin{lemma} \label{lem:Z-inv}
For integers $p,q, n\geq 1$,
\begin{align*}
E_{p,q,n} \subseteq F_{p,q,n}^{U(n)}.
\end{align*}
\end{lemma}

\begin{proof}
It is sufficient to show that any  $(i,j) \in [p] \times [q]$, the quadratic form $Z_{i,j}$ defined in Eq.~\eqref{eq:Zij} is invariant under the action of $U(n)$. This follows from the fact that $\sum_{k=1}^n (u z_{i})_k (\overline{u} z'_j)_k = \sum_{k=1}^n z_{i,k} z_{j,k}'$, for any $u \in U(n)$. 
\end{proof}

The idea of the proof of Theorem \ref{thm:charact-symm}, which we defer to the appendix, is to show that, for any $d \in \N$, the finite dimensional subspaces $E_{p,q,n}^{\leq d}$ and $F_{p,q,n}^{U(n), \leq d}$ are equal. Lemma \ref{lem:Z-inv} shows that $E_{p,q,n}^{\leq d} \subseteq F_{p,q,n}^{U(n), \leq d}$, for all $d \in \N$. We will establish in Corollary \ref{corol:dim} (in Appendix \ref{sec:charac}) that the dimensions of these subspaces coincide, which implies that the subspaces are the same, up to isomorphism. The theorem then results from the completeness of $F_{p,q,n}^{U(n)}$ and $E_{p,q,n}$.
In the following, we will use $F_{p,q,n}^{U(n)}$ and $E_{p,q,n}$ interchangeably. 

For an integer $d\in \N$, we define the (non normalized) \emph{maximally entangled state} $\Phi_{p+q,p+q,n}^{\leq d}$ on the subspace $E_{p+q,q+p,n}^{\leq d} = F_{p+q,q+p,n}^{\leq d}$ as:
\begin{align*}
\Phi_{p+q,p+q,n}^{\leq d} := \sum_{m=0}^d \frac{1}{m!} (Z_{1,q+1} + \ldots + Z_{p,q+p} + Z_{p+1,1} + \ldots + Z_{p+q,q})^m.
\end{align*}
In order to understand the name ``maximally entangled state'', it is useful to switch to the physicist bra-ket notation. The state $\Phi_{p+q,p+q,n}^{\leq d}$ is given by
\begin{align} \label{eq:MES}
|  \Phi_{p+q,p+q,n}^{\leq d} \rangle & = \sum_{m=0}^d \frac{1}{m!} \left(\sum_{k=1}^n \sum_{i=1}^p  a^{\dagger}_{k,i} a'^{\dagger}_{k,q+i}+ \sum_{k=1}^n  \sum_{j=1}^q a^{\dagger}_{k,p+j} a'^\dagger_{k,j}\right)^m |0\rangle
\end{align}
where $a_{k,i}^\dagger$ denotes the creation operator on mode $(k,i)$ of $F_A \cong F_{p+q,0,n}$, the operator $a_{k, j}'^\dagger$ denotes the creation operator on mode $(k,j)$ of $F_B \cong F_{0,p+q,n}$ and $|0\rangle$ is the vacuum state corresponding to the constant function $1$.

Let us define the weight and factorial of a table $M =(m_{k, \ell})_{k\in [n], \ell \in [p+q]}$ of nonnegative integers as:
\begin{align*}
|M| := \sum_{k\in [n], \ell \in [p+q]} m_{k, \ell} \quad \text{and} \quad M! := \prod_{k\in [n], \ell \in [p+q]} m_{k, \ell}!  .
\end{align*}
Expanding the product in Eq.~\eqref{eq:MES} and recalling that $(a^{\dagger}_{k,i} a_{k,q+i}'^{\dagger})^{m_{k,i}} |0\rangle= {m_{k,i}!} |m_{k,i}, m_{k,i} \rangle_{A_{k,i}, B_{k,q+i}}$ yields:
\begin{align*}
|  \Phi_{p+q,p+q,n}^{\leq d} \rangle & = \sum_{\substack{M \text{s.t.}\\ |M| \leq d}}   \frac{1}{M!} \left( \prod_{k \in [n], i\in [p]} (a^{\dagger}_{k,i} a_{k,q+i}'^{\dagger})^{m_{k,i}} \prod_{k \in [n], j\in [q]} (a^{\dagger}_{k,q+j} a_{k,j}'^{\dagger})^{m_{k,q+j}} \right)|0\rangle \\
&=  \sum_{\substack{M \text{s.t.}\\ |M| \leq d}}    \bigotimes_{k \in [n], i\in [p]} |m_{k,i}, m_{k,i} \rangle_{A_{k,i}, B_{k,q+i}} \bigotimes_{k \in [n], j\in [q]} |m_{k,q+j}, m_{k,q+j}\rangle_{A_{k,p+j},B_{k,j}}.
\end{align*}
This is the maximally entangled state between one copy of $F_{p,q,n}$ on modes $A_1, \ldots, A_p, B_1, \ldots, B_q$ and a copy of $F_{q,p,n} \cong F_{p,q,n}$ on modes $A_{p+1}, \ldots, A_{p+q}, B_{q+1}, \ldots, B_{p+q}$.

An important feature of this state is that it is invariant under the unitary group acting on $\mathfrak{B}(F_{p+q,p+q,n})$ as
\begin{align*}
U(n) & \to \mathfrak{B}(F_{p+q,p+q,n})\\
u & \mapsto W_u \otimes  W_{\overline{u}}= \big( \psi(z_1, \ldots, z_p; z_{p+1}, \ldots, z_{p+q}; z_1', \ldots, z_q'; z'_{q+1}, \ldots, z_{p+q}') \\
& \quad \quad \quad \quad \quad \quad \quad \quad \mapsto \psi(uz_1, \ldots, uz_p; uz_{p+1}, \ldots, uz_{p+q}; \overline{u}z_1', \ldots, \overline{u}z_q'; \overline{u}z'_{q+1}, \ldots, \overline{u}z_{p+q}')\big)
\end{align*}
where $W_u$ acts on modes $A_1, \ldots, A_p, B_1, \ldots, B_q$ and $W_{\overline{u}}$ acts on a copy of $F_{q,p,n} \cong F_{p,q,n}$ corresponding to modes $ B_{q+1}, \ldots, B_{p+q}, A_{p+1}, \ldots, A_{p+q}$ (where we recall that each factor $A_i$ consists of $n$ modes $A_{1,i}, \ldots, A_{n,i}$ and similarly for $B_j$).

\begin{lemma}\label{lem:purif}
 Let $d \geq 0$ be an integer, then the non normalized state $\Phi_{p+q,p+q,n}^{\leq d}$ corresponding to the maximally entangled state on the subspace $F_{p+q,p+q,n}^{\leq d}$ satisfies 
$$  \big(W_u \otimes  W_{\overline{u}} \big) \Phi_{p+q,p+q,n}^{\leq d} = \Phi_{p+q,p+q,n}^{\leq d}$$
for any $u \in U(n)$.
\end{lemma}

\begin{proof}
Observe that $Z_{i,q+i}$ and $Z_{p+j,j}$ are both invariant under the action of $W_u \otimes W_{\overline{u}}$, with the same argument as in Eq.~\eqref{eqn:Zinv}.
\end{proof}

\begin{rem}
The space $F_{p+q,p+q,n}^{\leq d}$ is left invariant under the action of $W_u \otimes W_{\overline{u}}$ for any $d$: this is because the change of variables $z \to uz$ and $z' \to \overline{u} z'$ preserves the total degrees in $z$ and $z'$. 
\end{rem}

The notion of invariance under the action of the unitary group $U(n)$ can be generalized to density operators, upon which $W_u$ acts by conjugation. 
\begin{defn} For $p,q,n \geq 1$, a density operator $\rho \in \mathfrak{S}(F_{p,q,n})$ is \emph{invariant under the action of the unitary group $U(n)$} if 
$$W_u \rho W_u^\dagger = \rho$$ for all $u \in U(n)$.
\end{defn}

Since we are dealing with infinite-dimensional operators, it is important to specify our topology. In this work, we will be working with the \textit{weak operator topology}.

\begin{defn}[Weak Operator Topology] A family $(\rho_m)_{m \in \N}$ of $\mathfrak{B}(F_{p,q,n})$ converges to $\rho \in \mathfrak{B}(F_{p,q,n})$ in the \emph{Weak Operator Topology} (WOT) if, for any $\psi \in F_{p,q,n}$, it holds that
\begin{align*}
\langle \psi, \rho_m \psi\rangle \to \langle \psi, \rho \psi\rangle,
\end{align*}
where we recall that
\begin{align*}
\langle \phi, \psi\rangle :=  \frac{1}{\pi^{n(p+q)}}\int \exp(-|z|^2 -|z'|^2) \overline{\phi(z,z')} \psi(z,z') \d z \d z'.
\end{align*}
\end{defn}

A crucial feature of symmetric subspace in the context of quantum information is that invariant density matrices always admit a purification in the symmetric subspace. The following result is a slight generalization of a result of Renner \cite{ren08} which holds when the invariance is with respect to the action of the symmetric group $S_n$ instead of $U(n)$.

\begin{reptheorem}{theo:purification}
For $p,q,n \geq 1$, any density operator $\rho \in \mathfrak{S}(F_{p,q,n})$ invariant under $U(n)$ admits a purification in $F_{p+q, p+q,n}^{U(n)}$.
\end{reptheorem}

\begin{proof}
Let denote by $\Pi_{p,q,n}^{\leq d}$ (or $\Pi^{\leq d}$ when the parameters $p,q,n$ are clear from context) the projector onto the finite-dimensional subspace $F_{p,q,n}^{\leq d}$. 
We note that the sequence of operators $\rho^{\leq d} := \Pi_{p,q,n}^{\leq d} \rho \Pi_{p,q,n}^{\leq d}$ converges in WOT to $\rho$ as $d$ tends to infinity.

We define the (non normalized) state $\psi^{\leq d} \in F_{p+q,p+q,n}^{\leq d}$ as follows:
\begin{align*}
\psi^{\leq d} := (\sqrt{\rho^d} \otimes \1) \Phi_{p+q,p+q,n}^{\leq d},
\end{align*}
where the operator $\sqrt{\rho^d}$ acts on $F_{p,q,n}$, \ie modes $A_1, \ldots, A_p, B_1, \ldots, B_q$, and the identity acts on modes $A_{p+1}, \ldots, A_{p+q}, B_{q+1}, \ldots, B_{p+q}$, \ie a space isomorphic to $F_{q,p,n}$.
Note that for any $d' \geq d$, it holds that 
\begin{align} 
\label{eqn:d1d2}
\psi^{\leq d} := (\sqrt{\rho^d} \otimes \1) \Phi_{p+q,p+q,n}^{\leq d'}.
\end{align}
It is well-known that the (non normalized) state $ \psi^{\leq d} = (\sqrt{\rho^d} \otimes \1) \Phi_{p+q,p+q,n}^{\leq d}$ is a purification of $\rho^{\leq d}$ since $\Phi_{p+q,p+q,n}^{\leq d}$ is the maximally entangled state on $F_{p,q,n} \otimes F_{p,q,n} \cong F_{p+q,p+q,n}$.
Furthermore, it satisfies  
\begin{align} \label{eqn:normalization}
\| \psi^{\leq d}\|^2 = \tr \rho^{\leq d}.
\end{align}
Moreover, since $W_u$ preserves the total degrees in $z$ and $z'$, we infer that $W_u$ commutes with the projector $\Pi_{p,q,n}^{\leq d}$. By assumption, it also commutes with $\rho$ and we conclude that $W_u \rho^{\leq d} W_u^\dagger = \rho^{\leq d}$ for all $u \in U(n)$ and $d \in \N$.
This implies:
\begin{align*}
 \big(W_u \otimes W_{\overline{u}} \big) \psi^{\leq d} &= \big(W_u \otimes  W_{\overline{u}} \big)  (\sqrt{\rho^d} \otimes \1)  \Phi_{p+q,p+q,n}^{\leq d}\\
 &= (\sqrt{\rho^d} \otimes \1) \big(W_u \otimes  W_{\overline{u}} \big)    \Phi_{p+q,p+q,n}^{\leq d} = \psi^{\leq d} 
\end{align*}
and establishes that $\psi^{\leq d} \in F_{p+q,p+q,n}^{U(n)}$ for any $d \in \N$.

We now claim that $\left( \psi^{\leq d} \right)_{d \in \N}$ forms a Cauchy sequence and therefore converges to
\begin{align*}
\psi = \lim_{d \to \infty} \psi^{\leq d},
\end{align*}
which is a purification of $\rho$ with $\psi \in F_{p+q,p+q,n}^{U(n)}$.
Indeed, for arbitrary integers $d_2 > d_1 \geq 0$, it holds that
\begin{align*}
\psi^{\leq d_2} - \psi^{\leq d_1} &= (\sqrt{\rho^{\leq d_2}} - \Pi^{\leq d_1} \sqrt{\rho^{\leq d_2}} \Pi^{\leq d_1})  \otimes \1) \Phi_{p+q,p+q,n}^{\leq d_2},
\end{align*}
where we used Eq.~\eqref{eqn:d1d2}. 
The operator $\sqrt{\rho^{\leq d_2}} - \Pi^{\leq d_1} \sqrt{\rho^{\leq d_2}} \Pi^{\leq d_1}$ is non negative and therefore there exists some operator $A \geq 0$ such that $\sqrt{A} = \sqrt{\rho^{\leq d_2}} - \Pi^{\leq d_1} \sqrt{\rho^{\leq d_2}} \Pi^{\leq d_1}$. 
Applying Eq.~\eqref{eqn:normalization} gives:
\begin{align*}
\|\psi^{\leq d_2} - \psi^{\leq d_1} \|^2 &= (\sqrt{A} \otimes \1) \Phi_{p+q,p+q,n}^{\leq d_2} = \tr A \\
&= \tr (\sqrt{\rho^{\leq d_2}} - \Pi^{\leq d_1} \sqrt{\rho^{\leq d_2}} \Pi^{\leq d_1} )^2 \\
&= \tr (\sqrt{\rho^{\leq d_2}} - \sqrt{\rho^{\leq d_1}})^2.
\end{align*}
By continuity of the function $\sigma \mapsto \sqrt{\sigma}$ on the set of subnormalized density matrices on $F_{p,q,n}$, we conclude that $\|\psi^{\leq d_2} - \psi^{\leq d_1} \|^2 \to 0$ when $d_1, d_2 \to \infty$, which means that the sequence $(\psi^{\leq d})_{d \in \N}$ is Cauchy and therefore converges in the Hilbert space $F_{p+q,p+q,n}^{U(n)}$. 
The limit of this sequence $\psi = \lim_{d\to \infty} \psi^{\leq d}$ is a purification of $\rho$ and belongs to $F_{p+q,p+q,n}^{U(n)}$.

\end{proof}


\section{Coherent states for $SU(p,q)/SU(p)\times SU(q) \times U(1)$}
\label{sec:CS}

In this section, we first review a construction due to Perelomov that associates a family of generalized coherent states to general Lie groups \cite{per72}, \cite{per86}. In this language, the standard Glauber coherent states are associated with the Heisenberg-Weyl group, while the atomic spin coherent states are associated with $SU(2)$. We then show that the symmetric subspace $F_{p,q,n}^{U(n)}$ is spanned by $SU(p,q)$ coherent states, where $SU(p,q)$ is the special unitary group of signature $(p,q)$ over $\C$:
\begin{align}
SU(p,q) := \left\{ A \in M_{p+q}(\C) \: : \: A \1_{p,q} A^\dagger =\1_{p,q} \quad \text{and} \quad \det A = 1\right\}
\end{align}
where $M_{p+q}(\C)$ is the set of $(p+q)\times (p+q)$-complex matrices and $\1_{p,q} = \1_{p} \oplus (-\1_q)$.

In Perelomov's construction, a \emph{system of coherent states of type} $(T, |\psi_0\rangle)$ where $T$ is the representation of some group $G$ acting on some Hilbert space $\mathcal{H} \ni |\psi_0\rangle$, is the set of states $\left\{|\psi_g\rangle \: : \: |\psi_g\rangle = T_g |\psi_0\rangle\right\}$ where $g$ runs over all the group $G$. One defines $H$, the \emph{stationary subgroup} of $|\psi_0\rangle$ as 
\begin{align*}
H := \left\{g \in G \: :Ê\: T_g |\psi_0\rangle = \alpha |\psi_0\rangle \, \text{for} \, |\alpha|=1Ê\right\},
\end{align*}
that is the group of $h \in G$ such that $|\psi_h\rangle $ and $|\psi_0\rangle$ differ only by a phase factor. When $G$ is a connected noncompact simple Lie group, $H$ is the maximal compact subgroup of $G$. 
In particular, for $G = SU(p,q)$, one has $H= SU(p,q) \cap U(p+q) = SU(p) \times SU(q)\times U(1)$ and the factor space $G/H$ corresponds to a Hermitian symmetric space of classical type (see \eg Chapter X of \cite{hel79}).  
The generalized coherent states are parameterized by points in $G/H$. For $G/H = SU(p,q)/SU(p)\times SU(q) \times U(1)$, the factor space is the set $\D$ of $p\times q$ matrices $\Lambda$ such that $\Lambda \Lambda^\dagger < \1_{p}$, \ie with the singular values strictly less than 1.
\begin{align*}
\D := \left\{ \Lambda \in M_{p,q}(\C) \: : \:\mathbbm{1}_p - \Lambda\Lambda^\dagger >0 \right\},
\end{align*}
where $A>0$ for a Hermitian matrix $A$ means that $A$ is positive definite. Here, $M_{p,q}(\C)$ is the set of $p\times q$ complex matrices. 
Remark that the set $\D$ could have alternatively been defined as follows:
\begin{align*}
\D = \left\{ \Lambda \in M_{p,q}(\C) \: : \:\mathbbm{1}_q - \Lambda^\dagger \Lambda>0 \right\}.
\end{align*}

Conversely, Cartan showed that $SU(p,q)$ is the automorphism group of $\D$ \cite{car35}.
\begin{theorem}[Cartan (1935)]
For $p,q \geq 1$, the automorphism group $\mathrm{Aut}\, \D$ of $\D$ is $SU(p,q)$, with group action is given by
\begin{align*}
SU(p,q) & \to \mathrm{Aut}\, \D \\
g =\left(\begin{smallmatrix}  A & B \\ C & D \end{smallmatrix}\right)  &\mapsto T_g := \left[ \Lambda \mapsto (A^T \Lambda + C^T) (B^T \Lambda+D^T)^{-1} \right]
\end{align*}
where $A, B, C, D$ are $p\times p, p\times q, q\times p, q \times q$ matrices, respectively, satisfying
\begin{align*}
A A^\dagger - B B^\dagger = \mathbbm{1}_p, \quad AC^\dagger = B D^\dagger, \quad D D^\dagger - C C^\dagger =\mathbbm{1}_q.
\end{align*}
\end{theorem}

We are now ready to define our coherent states for the noncompact Lie group $SU(p,q)$. 
\begin{defn}[$SU(p,q)$ coherent states] \label{defn:CS} For $p,q,n \geq 1$, the coherent state $\psi_{\Lambda,n}$ associated with $\Lambda \in \D$ is given by
 \begin{align}\label{eqn:CS-def}
 \psi_{\Lambda,n}(Z_{1,1}, \ldots, Z_{p,q}) = \det (1-\Lambda \Lambda^\dagger)^{n/2} \det \exp (\Lambda^T Z),
 \end{align}
 where $Z$ is the $p\times q$ matrix $\left[ Z_{i,j}\right]_{i \in [p], j\in [q]}$.
\end{defn}
Using Sylvester's determinant identity, we can equally write $ \psi_{\Lambda,n}(Z_{1,1}, \ldots, Z_{p,q}) = \det (1-\Lambda^\dagger \Lambda)^{n/2} \det \exp (\Lambda^T Z)$.

It will be sometimes useful to use the bra-ket notation, in which case the state $\psi_{\Lambda,n}$ will be denoted by $|\Lambda,n\rangle$. When the value of $n$ is omitted, it should be clear that it is taken equal to 1. In particular, we will write $|\Lambda \rangle$ instead of $|\Lambda,1\rangle$.

\begin{rem} The coherent states have a tensor product form in the sense that  
\begin{align*}
\psi_{\Lambda,n}=\psi_{\Lambda}^{\otimes n}, \quad \text{or} \quad |\Lambda, n\rangle = |\Lambda\rangle^{\otimes n}.
\end{align*}
Such states are referred to as \emph{identically and independently distributed} (i.i.d.) in the quantum information literature. 
\end{rem}
Using the identity $\det(\exp A) = \exp (\mathrm{tr} A)$, we immediately obtain an alternate expression for the coherent state, namely:
\begin{align*}
 \psi_{\Lambda,n}(Z_{1,1}, \ldots, Z_{p,q}) = \det (1-\Lambda \Lambda^\dagger)^{n/2}  \exp (\mathrm{tr} \Lambda^T Z)= \det (1-\Lambda \Lambda^\dagger)^{n/2}  \exp \left(\sum_{i,j} \Lambda_{i,j} Z_{i,j}\right).
\end{align*}
\begin{rem}
Perelomov also considered coherent states for the Lie group $SU(p,q)$, and labelled them with an index $k$ (see Eq.~(12.3.6) of \cite{per86}) corresponding to $n/(p+q)$ in our notations. Berezin \cite{ber75} gave a similar construction for symmetric spaces such as $\D$, and defined Planck's constant as $h := \frac{p+q}{n}$. The fact that we recover the classical setting in the limit $h\to \infty$ will be consistent with the property that generalized coherent states become orthogonal in the limit $n \to \infty$. 
\end{rem}

\begin{lemma}
The coherent states defined above belong to $F_{p,q,n}$.
\end{lemma}
\begin{proof}
We need to show that for any $\Lambda \in \D$, it is the case that $\|\psi_\Lambda\|^2 <\infty$. 
Let us consider the singular value decomposition of $\Lambda$, that is $\Lambda = U \Sigma V^{\dagger}$, where $U$ is a $p\times p$ unitary matrix, $\Sigma$ is a nonnegative $p \times q$ diagonal matrix and $V$ is a unitary $q\times q$ matrix. 
The norm of $\psi_\Lambda$ is
\begin{align*}
\|\psi_\Lambda\|^2 &= \frac{1}{\pi^{n(p+q)}}\int \exp(-|z|^2-|z'|^2 ) |\psi_\Lambda(z)|^2\mathrm{d}z \mathrm{d}z'\\
&= \frac{1}{\pi^{n(p+q)}} \det (1-\Sigma^2)^{n} \int \exp(-|z|^2 -|z'|^2 )\left|\exp \left(\sum_{i=1}^p \sum_{j=1}^q \sum_{k=1}^n \Lambda_{i,j} z_{k,i} z_{k,j}' \right) \right|^2  \mathrm{d}z \mathrm{d}z'\\
&= \frac{1}{\pi^{n(p+q)}} \det (1-\Sigma^2)^{n} \int \exp(-|z|^2 -|z'|^2 )\left|\exp \mathrm{tr} (M_z \Lambda M_{z'}^T )\right|^2 \mathrm{d}z \mathrm{d}z'\\
\end{align*}
where $M_z = \left[ z_{k,i}\right]_{k\in [n],i\in [p]}$ and $M_{z'} = \left[ z'_{k,j}\right]_{k\in [n],j\in [q]}$.
Next we observe that 
\begin{align*}
\mathrm{tr} (M_z \Lambda M_{z'}^T ) &= \mathrm{tr} (M_z U \Sigma V^{\dagger} M_{z'}^T ) \\
&= \sum_{i=1}^p \sum_{j=1}^q \sum_{k=1}^n \sum_{r=1}^{\min(p,q)}   z_{k,i} U_{i,r} \Sigma_{r,r} \overline{V}_{j,r} z_{j,k}' \\
&= \sum_{k=1}^n \sum_{r=1}^{\min(p,q)}   (U^T z)_{k,r} \Sigma_{r,r} (\overline{V} z')_{k,R} \\
&= \mathrm{tr} (M_{U^Tz} \Sigma M_{\overline{V} z'}^T )
\end{align*}
Performing the change of variable $z \to U^T z, z' \to \overline{V} z'$ in the integral above yields
\begin{align*}
\|\psi_\Lambda\|^2  = \|\psi_\Sigma\|^2
\end{align*}
since the measure $\exp(-|z|^2) \d z$ is invariant under a unitary change of variable.
Now, the state $\psi_\Sigma$ with the diagonal matrix $\Sigma$ is simply the tensor product of $n\times \min(p,q)$ two-mode squeezed vacuum states:
\begin{align*}
\psi_\Sigma(z) &= \prod_{i=1}^{\min(p,q)} (1-|\Sigma_{i,i}|^2)^{n/2} \exp(\Sigma_{i,i} Z_{i,i})\\
&= \prod_{i=1}^{\min(p,q)} \prod_{k=1}^n  \left((1-|\Sigma_{i,i}|^2)^{1/2} \exp(\Sigma_{i,i} z_{k,i} z'_{k,i}) \right),
\end{align*}
each one of which being correctly normalized, with squeezing parameter $|\Sigma_{i,i}| <1$.
\end{proof}

By construction, the $SU(p,q)$ coherent states are invariant under the action of the unitary group $U(n)$ on $F_{p,q,n}$, which means that they belong to the symmetric subspace $F_{p,q,n}^{U(n)}$. We show that they in fact span this symmetric subspace. 

\begin{theorem}
For $p,q, n\geq 1$, the $SU(p,q)$ coherent states span the symmetric subspace $F_{p,q,n}^{U(n)}$.
\end{theorem}

The proof is similar to the case of $SU(d)$ coherent states (see \eg \cite{GW00}).
\begin{proof}
The idea is to show that for any vector $m=(m_{1,1}, \ldots, m_{p,q}) \in \N^{pq}$, the polynomial $\prod_{i \in [p]} \prod_{j \in [q]} Z_{i,j}^{m_{i,j}}$ is in the span of the coherent states. 
To see this, consider the function 
\begin{align*}
f: \quad& \D \to \C[[Z_{1,1}, \ldots, Z_{p,q}]]\\
&\Lambda \mapsto \exp\left( \sum_{i \in [p]} \sum_{j\in [q] } \Lambda_{i,j} Z_{i,j} \right). 
\end{align*}
Evaluating the derivative $\prod_{i \in [p]} \prod_{j \in [q]} \frac{\partial^{m_{i,j}}}{\partial \Lambda_{i,j}^{m_{i,j}}} f$ in $\Lambda=0$ gives the monomial $\prod_{i \in [p]} \prod_{j \in [q]} Z_{i,j}^{m_{i,j}}$, which proves that this monomial lies in the algebra of the coherent states. 
\end{proof}

We are now ready to define the unitary representation of the noncompact group $SU(p,q)$ on the symmetric subspace $F_{p,q,n}^{U(n)}$, which was also considered by Perelomov \cite{per86}.
\begin{defn}[Representation of $SU(p,q)$]\label{def:rep}
The generalized unitary group $SU(p,q)$ admits the following discrete series unitary representation on the symmetric subspace $F_{p,q,n}^{U(n)}$:
\begin{align*}
T: \quad SU(p,q) & \to \mathfrak{B}(F_{p,q,n}^{U(n)})\\
g =\left(\begin{smallmatrix}  A & B \\ C & D \end{smallmatrix}\right) & \mapsto T_g := \big[ |\Lambda, n\rangle \mapsto |\Lambda_g, n\rangle\big],
\end{align*}
where the matrix $\Lambda_g$ is given by $\Lambda_g =  (A^T \Lambda + C^T) (B^T \Lambda+D^T)^{-1}$.
\end{defn}

A characterization of a \textit{discrete series representation} $T$ is that for any $\Lambda_1, \Lambda_2 \in \D$, the function $g \mapsto \langle \Lambda_1,n | T_g|\Lambda_2,n\rangle$ is square-integrable with an invariant measure on $SU(p,q)$. It was shown by Perelomov \cite{per86} that this is indeed the case here.


The overlap, or fidelity, between generalized coherent states admits a simple formula.

\begin{theorem}\label{thm:overlap}
Let $\Lambda_1, \Lambda_2 \in \D$. The fidelity $F(\Lambda_1, \Lambda_2) := |\langle \Lambda_1 | \Lambda_2\rangle|^2$ between the two $SU(p,q)$ coherent states $|\Lambda_1\rangle$ and $|\Lambda_2\rangle$ is given by
\begin{align*}
|\langle \Lambda_1 | \Lambda_2\rangle|^2 = \frac{\det(\1_p- \Lambda_1\Lambda_1^\dagger) \det(\1_p- \Lambda_2\Lambda_2^\dagger)}{|\det(\1_p- \Lambda_1\Lambda_2^\dagger)|^2}.  
\end{align*}
\end{theorem}
The generalization to arbitrary values of $n$ is immediate since $|\langle \Lambda_1,n | \Lambda_2,n\rangle|^2 = (|\langle \Lambda_1,1 | \Lambda_2,1\rangle|)^{2n}$.
Let us first establish two simple results.

\begin{lemma}\label{lem:overlap0}
For $\Lambda \in \D$, we have 
\begin{align*}
|\langle 0 |\Lambda\rangle|^2 = \det (1-\Lambda \Lambda^\dagger).
\end{align*}
\end{lemma}
\begin{proof}
This is an immediate consequence of the expansion of the state $|\Lambda\rangle$ in the Fock basis with Eq.~\eqref{eqn:CS-def}.
\end{proof}

\begin{lemma}
Let $\Sigma_1, \Sigma_2 \in \D$ be diagonal matrices with nonnegative entries, then
\begin{align*}
\langle \Sigma_1 | \Sigma_2\rangle = \frac{\det(1-\Sigma_1^2)^{1/2} \det(1-\Sigma_2^2)^{1/2}}{\det(1-\Sigma_1\Sigma_2)}.
\end{align*}
\end{lemma}

\begin{proof}
If $\Sigma_i = \mathrm{diag}(\sigma_{i,1}, \cdots, \sigma_{i,p})$ for $i \in \{1,2\}$, then $|\Sigma_i\rangle = \bigotimes_{j=1}^p |\sigma_{i,j}\rangle$ where the state $|\sigma\rangle := \sqrt{1-\sigma^2}\sum_{k=0}^\infty \sigma^k |k,k\rangle$ is a two-mode squeezed vacuum state for $\sigma \in [0,1)$ and where we omit the tensor product with the remaining $(q-p)$ vacuum modes. 
The overlap between two such states is straightforward to compute:
\begin{align*}
\langle \sigma_1 | \sigma_2\rangle = \sqrt{1-\sigma_1^2} \sqrt{1-\sigma_2^2} \sum_{k=0}^\infty (\sigma_1\sigma_2)^k = \frac{ \sqrt{1-\sigma_1^2} \sqrt{1-\sigma_2^2}}{1-\sigma_1 \sigma_2},
\end{align*}
which yields the desired result.
\end{proof}

\begin{proof}[Proof of Theorem \ref{thm:overlap}]

Let us prove the general case for arbitrary $p\leq q$. 

For matrices $\Lambda \in \D$ and $g = \left(\begin{smallmatrix}A & B\\C & D\end{smallmatrix}\right) \in SU(p,q)$, let us denote $\Lambda_g := (A^T \Lambda + C^T)(B^T \Lambda + D^T)^{-1}$.
Since the action of $g$ is unitary, we have $\langle \Lambda_g | \Lambda_g'\rangle = \langle \Lambda | \Lambda'\rangle$ for any matrices $\Lambda, \Lambda' \in \D$ and $g \in SU(p,q)$.

Let $\Lambda_1 = U \Sigma V^\dagger$ be the singular value decomposition of $\Lambda_1$ with $U \in U(p)$ and $V\in U(q)$. For $g =  \left(\begin{smallmatrix} \overline{U} & 0\\0 & \overline{V} \end{smallmatrix}\right) \in SU(p,q)$, we obtain $(\Lambda_1)_g = U^\dagger \Lambda_1 V = \Sigma$. This yields:
\begin{align}
\langle \Lambda_1 | \Lambda_2 \rangle = \langle \Sigma |Ê\Lambda\rangle
\end{align}
with $\Lambda := (\Lambda_2)_gU^\dagger \Lambda_2 V$.
Since $\Sigma$ is a $p\times q$ diagonal matrix with nonnegative entries, there exists a $p\times q$ diagonal matrix $R$ such that $\Sigma=\tanh R$. We further define the $p\times q$ matrix $S = \sinh R$ as well as the $p\times p$ matrix $C = \cosh R$ and the $q\times q$ matrix $\tilde{C}$ which is simply the matrix $C$ padded with the appropriate number of ones on the diagonal. With these notations, the matrix $g' = \left[ \begin{smallmatrix} C & -S\\ -S^T & \tilde{C} \end{smallmatrix} \right]$ belongs to $SU(p,q)$ and $\Sigma_{g'}= -(C \Sigma - S) (S^T T - \tilde{C})^{-1}=0$. 
By unitarity of the action of $g'$, we obtain
\begin{align*}
|\langle \Sigma |\Lambda\rangle|^2 = |\langle \Sigma_{g'} |\Lambda_{g'}\rangle|^2 = |\langle 0 | \Lambda_{g'}\rangle^2= \det (\1-\Lambda_{g'} \Lambda_{g'}^\dagger)
\end{align*}
where the last equality follows from Lemma \ref{lem:overlap0}.
This determinant can be computed explicitly since 
\begin{align*}
\Lambda_{g'} = -(C\Lambda -S)(S^\dagger\Lambda-\tilde{C})^{-1}, \quad
\Lambda_{g'}^\dagger = -((S^\dagger \Lambda-\tilde{C})^{\dagger})^{-1}(C\Lambda-S)^{\dagger}
\end{align*}
and
\begin{align*}
\det(\1- \Lambda_{g'}^\dagger\Lambda_{g'}) &= \det (1-((S^\dagger \Lambda-\tilde{C})^{\dagger})^{-1}(C\Lambda-S)^{\dagger}(C\Lambda -S)(S^\dagger \Lambda-\tilde{C})^{-1}\\
&=\frac{\det((S^\dagger \Lambda-\tilde{C})^{\dagger}(S^\dagger \Lambda-\tilde{C}) -(C\Lambda-S)^{\dagger}(C\Lambda -S))}{\det (S^\dagger \Lambda-\tilde{C})^{\dagger}(S^\dagger \Lambda-\tilde{C})}.
\end{align*}
Denoting by $A$ the numerator and $B$ the denominator of this last expression, we get:
\begin{align*}
A &= \det (\tilde{C}^\dagger \tilde{C} - S^\dagger S  -\Lambda^\dagger (C^\dagger C - S S^\dagger )\Lambda + \Lambda^\dagger(C^\dagger S-S \tilde{C} ) + (S^\dagger C- \tilde{C}^\dagger S^\dagger) \Lambda)\\
&=\det (\1_q - \Lambda^\dagger \Lambda)
\end{align*}
where we exploited that $\tilde{C}^\dagger \tilde{C} - S^\dagger S = \1_q$ and $C^\dagger C - S S^\dagger = \1_p$.
In order to analyze the term $B$, it is useful to recall a matrix identity that holds for hyperbolic trigonometric functions:
\begin{align*}
\tilde{C}  (1_q - \Sigma^\dagger \Sigma) \tilde{C}= \1_q,
\end{align*}
which implies $ \det(\tilde{C}^2 )\det (\1_q -  \Sigma^\dagger \Sigma) = 1$. 
Recalling that $S = \Sigma \tilde{C}$, we obtain
\begin{align*}
B &= \det (S^\dagger \Lambda-\tilde{C})(S^\dagger \Lambda-\tilde{C})^{\dagger}\\
&= \det( S^\dagger  \Lambda \Lambda^\dagger S  + \tilde{C}\tilde{C}^\dagger -S^\dagger \Lambda \tilde{C}^\dagger - \tilde{C}\Lambda^\dagger S)\\
&= \det( \tilde{C} \Sigma^\dagger  \Lambda \Lambda^\dagger \Sigma \tilde{C}  + \tilde{C}\tilde{C}^\dagger -\tilde{C} \Sigma^\dagger \Lambda \tilde{C}^\dagger - \tilde{C}\Lambda^\dagger  \Sigma \tilde{C})\\
&= \det(\tilde{C}^2 ) \det( \Sigma^\dagger  \Lambda \Lambda^\dagger \Sigma   + \1_q- \Sigma^\dagger \Lambda -\Lambda^\dagger  \Sigma )\\
&= \frac{\det(\1_q -  \Sigma^\dagger \Lambda)(\1_q - \Lambda^\dagger  \Sigma)}{\det (\1_q -  \Sigma^\dagger \Sigma)}
\end{align*}
which finally gives
\begin{align}\label{eqn:almost}
\det(\1_q- \Lambda_{g'}^\dagger\Lambda_{g'}) &= \frac{\det(\1_q-\Lambda^\dagger\Lambda) \det(\1_q-\Sigma^\dagger \Sigma)}{\det (\1_q - \Sigma^\dagger \Lambda)\det (\1_q  - \Lambda^\dagger \Sigma)}.
\end{align}

For two general matrices, $\Lambda_1, \Lambda_2 \in \D$, the singular value decomposition $\Lambda_1 = U \Sigma V^\dagger$ yields:
\begin{align*}
|\langle \Lambda_1 |\Lambda_2\rangle|^2 &= |\langle \Sigma |U^\dagger \Lambda_2 V\rangle|^2\\
&=\frac{\det(\1_q-V^\dagger \Lambda_2^\dagger  \Lambda_2 V) \det(\1_q-\Sigma^\dagger \Sigma)}{\det (\1_q - \Sigma^\dagger U^\dagger \Lambda_2 V)\det (\1_q  - V^\dagger \Lambda_2^\dagger U\Sigma)}\\
&=\frac{\det(\1_q- \Lambda_2^\dagger  \Lambda_2 ) \det(\1_q- V\Sigma^\dagger \Sigma V^\dagger)}{\det (\1_q - V \Sigma^\dagger U^\dagger \Lambda_2)\det (\1_q  -  \Lambda_2^\dagger U\Sigma V^\dagger)}\\
&=\frac{\det(\1_q- \Lambda_2^\dagger  \Lambda_2 ) \det(\1_q-\Lambda_1^\dagger \Lambda_1)}{\det (\1_q - \Lambda_1^\dagger \Lambda_2)\det (\1_q  -  \Lambda_2^\dagger \Lambda_1)}
\end{align*}
where the first equality follows from applying $g = \left[ \begin{smallmatrix} \overline{U} & 0\\ 0 & \overline{V} \end{smallmatrix} \right]$ to both states, and where we used Eq.~\eqref{eqn:almost} in the second equality.
\end{proof}


The main feature of a family of coherent states is that they resolve the identity. This is the case with the $SU(p,q)$ coherent states introduced above. 

\begin{reptheorem}{thm:resol}[Resolution of the identity]
For integers $p,q \geq 1$ and $n \geq p+q$, the coherent states resolve the identity over the symmetric subspace $F_{p,q,n}^{U(n)}$:
\begin{align*}
\int_{\D} |\Lambda,n\rangle \langle \Lambda,n| \mathrm{d}\mu_{p,q,n}(\Lambda) = \mathbbm{1}_{F_{p,q,n}^{U(n)}}, 
\end{align*}
where $\mathrm{d}\mu_{p,q,n}(\Lambda)$ is the invariant measure on $\D$ given by
\begin{align}\label{eqn:mu}
\mathrm{d}\mu_{p,q,n} (\Lambda) := C_n \d \mu_{p,q}(\Lambda) = C_n [\det(\mathbbm{1}_p - \Lambda \Lambda^\dagger)]^{-(p+q)} \prod_{i=1}^p \prod_{j=1}^{q} \mathrm{d} \mathfrak{R}(\Lambda_{i,j}) \mathrm{d} \mathfrak{I}(\Lambda_{i,j}), 
\end{align}
with the normalization constant
\begin{align}\label{eqn:cn}
C_n = \frac{1}{\pi^{pq}} \frac{(n-q)! \ldots (n-1)!}{(n-p-q)! \ldots (n-p-1)!}= \frac{1}{\pi^{pq}} \prod_{i=0}^{q-1} \frac{(n-q+i)!}{(n-p-q+i)!}, 
\end{align}
where $\mathfrak{R}(\Lambda_{i,j})$ and $\mathfrak{I}(\Lambda_{i,j})$ refer respectively to the real and imaginary parts of $\Lambda_{i,j}$. This operator equality is to be understood for the weak operator topology.
\end{reptheorem}
In the representation theory of Lie groups \cite{har51}, \cite{har56}, the constant $C_n$ appearing in the resolution of the identity (Theorem \ref{thm:resol}) is called the \emph{formal degree} of the representation of the group $SU(p,q)$ in $F_{p,q,n}^{U(n)}$. 

\begin{proof}
We refer to Perelomov (\cite{per86}, chapter 12) for the claim that $\mathrm{d}\mu_{p,q,n}(\Lambda)$ is the invariant measure on $\D$. The fact that the integral is proportional to the identity on the symmetric subspace is a consequence of a version of Schur's lemma applying to general unimodular groups such as $SU(p,q)$ with a square-integrable representation, which shows the existence of a constant $d_F$, called the degree of the representation, such that
\begin{align*}
\int_{SU(p,q)} \langle \psi |T_g | \phi\rangle \overline{\langle \psi' |T_g|\phi'\rangle} \d\mu(g)= \frac{1}{d_F} \langle \psi|\psi'\rangle \overline{\langle \phi|\phi'\rangle},
\end{align*}
for all $|\psi\rangle, |\phi\rangle, |\psi'\rangle, |\phi'\rangle \in F_{p,q,n}^{U(n)}$, where $\d\mu$ is an invariant measure on $SU(p,q)$ \cite{gaa73}.
\end{proof}

The resolution of the identity of Theorem \ref{thm:resol} implies that the family of generalized coherent states defines a generalized measurement on the symmetric subspace. It is therefore possible to associate to any state $\rho$ acting on $F_{p,q,n}^{U(n)}$ a generalized Husimi function $Q_\rho: \D \to R$ corresponding to the probability density function of obtaining a particular outcome $\Lambda$ when measuring the state $\rho$ with this specific measurement.
\begin{defn}[Generalized Husimi function]
Any state $\rho \in \mathfrak{S}(F_{p,q,n}^{U(n)})$ gives rise to a probability density function
\begin{align}
Q_{\rho}(\Lambda) = C_n \frac{\tr ( \rho |\Lambda\rangle \langle \Lambda|) }{\det(\mathbbm{1}_p - \Lambda \Lambda^\dagger)^{(p+q)}}
\end{align}
corresponding to the probability density function of obtaining the outcome $\Lambda$ when measuring $\rho$ with the measurement corresponding to the resolution of the identity of Theorem \ref{thm:resol}.
\end{defn}

We conclude this section by noting that Puri \cite{pur94} studied a family of $SU(p,q)$ coherent states similar to those considered in the present work, but in the case where $n=1$. In contrast, we are more interested in the regime where $n \geq p+q$, which is required for obtaining the resolution of the identity of Theorem \ref{thm:resol} as well as our Gaussian de Finetti theorem of Section \ref{sec:dF}.


\section{Phase-space representation of $SU(p,q)$ coherent states}
\label{sec:gauss}

The $SU(p,q)$ coherent states defined in the previous section are i.i.d.~states, \ie $|\Lambda, n\rangle = |\Lambda, 1\rangle^{\otimes n}$, and we will therefore restrict our attention to the case $n=1$ in this section, and simply write $|\Lambda\rangle$ instead of $|\Lambda,1\rangle$ when there is no ambiguity.
In particular, the Fock space we consider is $F(H_A \oplus H_B) = F(\C^p \oplus \C^q)$. The annihilation operators of $F(H_A)$ are denoted by $z_1, \cdots, z_p$, those of $F(H_B)$ by $z'_1, \cdots, z'_q$.

Since these generalized coherent states are Gaussian states (\ie their Wigner function is Gaussian), they are entirely characterized by the first two moments of their Wigner function \cite{FOP05},\cite{WPG12}. By symmetry, the first moment is null, which means that we are only concerned in computing the covariance matrix $\Gamma_{\Lambda}$ of the $(p+q)$-mode state $|\Lambda,1\rangle$.

We need to fix an ordering of the $2(p+q)$ quadratures corresponding to $|\Lambda,1\rangle$: we choose $(x_{1}, \cdots, x_{p}, y_{1}, \cdots, y_{p}, x'_{1}, \cdots, x'_{p}, y'_{1}, \cdots, y'_{q})$, with quadrature operators defined by:
\begin{align*}
x_{i} = z_{i} + \frac{\partial}{\partial z_{i}}, \quad &y_{i} = i\left(z_{i}- \frac{\partial}{\partial z_{i}}\right) & \text{for} \, i \in [p],\\
x'_{j} = z'_{j} + \frac{\partial}{\partial z_{j}}, \quad & y'_{j} = i\left(z'_{j}- \frac{\partial}{\partial z'_{j}}\right) & \text{for}\, j \in [q],
\end{align*}
where $z_{i}$ (resp.~$z'_{j}$) corresponds to the \emph{creation} operator associated with mode $i$ of $F(H_A)$ (resp.~$j$ of $F(H_B)$) and $\frac{\partial}{\partial z_{i}}$, $\frac{\partial}{\partial z'_{j}}$ correspond to the \emph{annihilation} operators associated with these modes. Note that we used the convention $\hbar = 2$ when defining the quadrature operators, similarly as in \cite{WPG12}.

Recall that the symplectic group $\Sp(2(p+q),\R)$ is the group of $2(p+q) \times 2(p+q)$ real matrices $S$ preserving the symplectic form $\Omega = \left[ \begin{smallmatrix} & \1_{p} \\ -\1_{p} & \end{smallmatrix} \right] \oplus \left[ \begin{smallmatrix} & \1_{q} \\ -\1_{q} & \end{smallmatrix} \right]$:
\begin{align}
\Sp(2(p+q),\R) := \left\{ S \in M_{2(p+q)}(\R) \: :\: S \Omega S^T = \Omega\right\}.
\end{align}

For complex matrices $X \in M_{p}(\C), Y \in M_{p,q}(\C)$, we define $\cS(X) \in \R^{2p \times 2p}$ and $\cT(Y) \in \R^{2p \times 2q}$ as:
\begin{align} \label{eq:ST}
\cS(X) := \begin{bmatrix}
\mathfrak{R}(X) & -\mathfrak{I}(X) \\
\mathfrak{I}(X) &  \mathfrak{R}(X)
\end{bmatrix}, \quad
\cT(Y):= \begin{bmatrix}
\mathfrak{R}(Y) & \mathfrak{I}(Y) \\
\mathfrak{I}(Y) & - \mathfrak{R}(Y)
\end{bmatrix}.
\end{align}
It is easy to verify that if $X$ is Hermitian, $X = X^\dagger$, then $\cS(X)$ is symplectic. Similarly, if $U$ is unitary, then $\cS(U)$ is also symplectic.

With these notations, we can compute the covariance matrix of an arbitrary coherent state $|\Lambda\rangle$.
\begin{theorem} \label{lem:CM}
Let $\Lambda \in \D$, the covariance matrix $\Gamma_{\Lambda}$ of the coherent state $|\Lambda,1\rangle$ is given by
\begin{align}
\Gamma_\Lambda = \begin{bmatrix}
\cS(\1_p + \Lambda\Lambda^\dagger) (\1_p- \Lambda \Lambda^\dagger)^{-1}) & \cT(2 \Lambda (\1_q - \Lambda^\dagger \Lambda)^{-1})\\
\cT(2 \Lambda (\1_q - \Lambda^\dagger \Lambda)^{-1})^T & \cS((\1_q + \Lambda^\dagger\Lambda) (\1_q-  \Lambda^\dagger\Lambda)^{-1})
\end{bmatrix}.
\end{align} 
\end{theorem} 

\begin{proof}
Let us apply the singular value decomposition for the matrix $\Lambda$, namely $\Lambda = U \Sigma V^\dagger$. 
It is straightforward to compute the covariance matrix associated with the $p\times q$ diagonal matrix $\Sigma = \mathrm{diag}(\sigma_1, \cdots, \sigma_p)$ since the state $|\Sigma,1\rangle$ is simply the tensor product of $p$ two-mode squeezed vacuum states with squeezing parameters $\Sigma_{1}$ to $\Sigma_{p}$ tensored with $(q-p)$ vacuum modes:
\begin{align}
\Gamma_\Sigma = \begin{bmatrix}[cc : cc ]
\nu &\cdot & \beta  &\cdot  \\
\cdot & \nu &\cdot   & -\beta  \\
 \hdashline
\beta^T &\cdot  & \tilde{\nu} &   \cdot \\
\cdot & -\beta^T & \cdot & \tilde{\nu}  \\
\end{bmatrix}
\end{align}
where $\nu = \mathrm{diag}\left(\frac{1+\sigma_1^2}{1-\sigma_1^2}, \cdots, \frac{1+\sigma_p^2}{1-\sigma_p^2}\right) = (\1_p + \Sigma\Sigma^\dagger) (\1_p- \Sigma\Sigma^\dagger)^{-1}$ is a $p\times p$ matrix, $\tilde{\nu} = \nu \oplus \1_{q-p}$ is a $q \times q$ matrix and $\beta =  2\Sigma (\1_q - \Sigma^\dagger \Sigma)^{-1}$ is a $p \times q$ matrix.

Then, we note that if $\Lambda$ and $\Sigma$ are related through $\Lambda = U \Sigma V^\dagger$, then their covariance matrices $\Gamma_\Lambda$ and $\Gamma_\Sigma$ satisfy
$\Gamma_{\Lambda} = [ \cS(U) \oplus \cS(\overline{V})] \Gamma_{\Sigma}  [ \cS(U) \oplus \cS(\overline{V})] ^T$.
A straightforward calculation yields:
\begin{align*}
\Gamma_\Lambda = 
\begin{bmatrix}
\cS(U \nu  U^\dagger) & \cT(U \beta V^\dagger)\\
\cT(U \beta V^\dagger)^T & \cS(V \tilde{\nu} V^\dagger)
\end{bmatrix}.
\end{align*}
This can be further simplified by noting that:
\begin{align*}
U \nu U^\dagger &=  (\1_p + \Lambda\Lambda^\dagger) (\1_p- \Lambda \Lambda^\dagger)^{-1} \\
V \nu V^\dagger &= (\1_q + \Lambda^\dagger\Lambda) (\1_q-  \Lambda^\dagger\Lambda)^{-1} \\
U \beta V^\dagger &= 2 \Lambda (\1_q - \Lambda^\dagger \Lambda)^{-1} = 2(\1_p -\Lambda \Lambda^\dagger)^{-1} \Lambda
\end{align*}
which concludes the proof.
\end{proof}

Since the coherent states are Gaussian, it is clear that the action of the group $SU(p,q)$ on matrices $\Lambda \in \D$ could alternatively be studied as the action of a subgroup of $\Sp(2(p+q))$ on phase space. A natural question is therefore to determine the mapping $s$ between the generalized unitary group $SU(p,q)$ and the symplectic group $\Sp(2(p+q),\R)$:
\begin{align*}
s: SU(p,q) \to \Sp(2(p+q),\R), g \mapsto S(g)
\end{align*}
satisfying $\Gamma_{\Lambda_g} = S(g) \Gamma_{\Lambda} S(g)^T$.

\begin{theorem}\label{thm:symp}
The map 
\begin{align}
s : SU(p,q) &\to \Sp(2(p+q),\R) \nonumber\\
g=\left[\begin{smallmatrix}A & B \\ C & D \end{smallmatrix}\right] & \mapsto s(g) = \left[\begin{smallmatrix}\cS(A) & \cT(B) \\ \cT(\overline{C}) & \cS(\overline{D}) \end{smallmatrix}\right]
\end{align}
is a homorphism. Moreover, for any $\Lambda \in \D$ and $g \in SU(p,q)$, it holds that $\Gamma(\Lambda_g) = s(g) \Gamma(\Lambda) s(g)^\dagger$.
\end{theorem}
In other words, the homomorphism $s$ describes how the group element $g \in SU(p,q)$ acts in phase space.

\begin{lemma}
Let $g=\left[\begin{smallmatrix}A & B \\ C & D \end{smallmatrix}\right] \in SU(p,q)$. Then the matrix $G=\left[\begin{smallmatrix}\cS(A) & \cT(B) \\ \cT(\overline{C}) & \cS(\overline{D}) \end{smallmatrix}\right]$ is symplectic, where $\cS$ and $\cT$ are defined as in Eq.~\eqref{eq:ST}.
\end{lemma}

\begin{proof}
Since $g \in SU(p,q)$, the following identities hold:
\begin{align*}
A A^\dagger - B B^\dagger = \1_p, \quad CC^\dagger - DD^\dagger = \1_q, \quad AC^\dagger = B D^\dagger.
\end{align*}
Our goal is to establish that $G \Omega G^T = \Omega$, with $\Omega = \Omega_p \oplus \Omega_q$ and $\Omega_p = \left[ \begin{smallmatrix} 0 & \1_p \\ -\1_p & 0\end{smallmatrix} \right]$.
It is straightforward to check that 
\begin{align*}
G\Omega G^T &=
\begin{bmatrix} 
\fI(A A^\dagger) - \fI(B B^\dagger)
& \fR(A A^\dagger)- \fR(B B^\dagger)
& \fI(AC^\dagger) - \fI(BD^\dagger)
& -\fR(AC^\dagger) + \fR(BD^\dagger)\\
- \fR(A A^\dagger) +\fR(B B^\dagger)
& \fI(A A^\dagger)- \fI(B B^\dagger)
& -\fR(AC^\dagger) + \fR(BD^\dagger)
& -\fI(AC^\dagger) + \fI(BD^\dagger)\\
\fI(CA^\dagger) - \fI(DB^\dagger)
& \fR(CA^\dagger)  - \fR(DB^\dagger)
& \fI(CC^\dagger)- \fI(DD^\dagger)
& - \fR(CC^\dagger)+ \fR(DD^\dagger)\\
\fR(CA^\dagger) - \fR(DB^\dagger)
&  -\fI(CA^\dagger)  + \fI(DB^\dagger)
& \fR(CC^\dagger)- \fR(DD^\dagger)
& \fI(CC^\dagger)- \fI(DD^\dagger)
\end{bmatrix}\\
&= 
\begin{bmatrix} 
0& \1_p &0&0\\
- \1_p &0&0&0\\
0&0&0& \1_q\\
0&0 & -\1_q &0 
\end{bmatrix}.
\end{align*}
\end{proof}

\begin{proof}[Proof of Theorem \ref{thm:symp}]
For any two matrices $g_1, g_2 \in SU(p,q)$, one can verify explicitly in a straightforward manner that $s(g_1 g_2) = s(g_1) s(g_2)$, which establishes that the map $s$ is compatible with multiplication and that $s$ is a homomorphism.

The group $SU(p,q)$ is generated by elements of the form $g_{U,V} = \left[ \begin{smallmatrix} U & 0 \\ 0 & V\end{smallmatrix} \right]$ for arbitrary unitary matrices $U \in U(p), V \in U(q)$ satisfying $\det (UV)=1$ and elements of the form $h_{R} = \left[ \begin{smallmatrix} C & S \\ S^T & \tilde{C}\end{smallmatrix} \right]$ for $C = \cosh R$ and $R$ an arbitrary real diagonal $p\times p$ matrix, $S$ the $p\times q$ matrix equal to $\sinh R$ on its $p$ left columns and 0 elsewhere, and $\tilde{C}$ the $q\times q$ matrix equal to $C$ on the upper-left $p\times p$ submatrix, with diagonal padded by 1s.
Indeed, given any pair of matrices $\Lambda_1, \Lambda_2 \in \D$, and their singular value decomposition $\Lambda_1 = U_1 \Sigma_1 V_1^\dagger$, $\Lambda_2 = U_2 \Sigma_2 V_2^\dagger$, the group element $g_{1 \to 2} := g_{U_2^T, V_2^T} \cdot h_{R} \cdot g_{\overline{U_1}, \overline{V_1}}$ maps $\Lambda_1$ to $\Lambda_2$ if $R = \artanh(\Sigma_2) -\artanh(\Sigma_1)$.

For this reason, it is sufficient to verify that $s$ correctly maps elements of the form $g_{U,V}$ and $h_{R}$ to the correct symplectic matrices. 
First, $g_{U,V}$ maps $\Lambda$ to $U^T \Lambda \overline{V}$. In particular, the state $\exp \left(\sum_{i,j} \Lambda_{i,j} z_i z_j'\right)$ is mapped to 
\begin{align*}
\exp \left(\sum_{i,j,k,\ell} U_{j,i} \Lambda_{j,k} \overline{V}_{k, \ell} z_i z_\ell'\right) =\exp \left(\sum_{j,k}  \Lambda_{j,k}  \left( \sum_{i} U_{j,i} z_i\right) \left(\sum_{\ell}  \overline{V}_{k, \ell} z_\ell'\right) \right).
\end{align*}
In other words, the creations operators of $F_A$ are transformed according to $\vec{z} \to U \vec{z}$, while those of $F_B$ are according to $\vec{z}' \to \overline{V} \vec{z'}$. The corresponding symplectic matrix is then given by $s(g_{U,V}) = \cS(U) \oplus \cS(\overline{V})$, see for instance \cite{WPG12}.
Similarly, the element $h_{R}$ simply applies $p$ two-mode squeezing operators in parallel and the corresponding action in phase space is given by $s(h_{R})$, see also \cite{WPG12}.
\end{proof}

\begin{rem}
According to Theorem \ref{thm:symp}, for any $g \in SU(p,q)$, the unitary transformation $T_g$ acts as a Gaussian transformation on $F_{p,q,n}^{U(n)}$. It is therefore possible to extend the representation $T$ to the Fock space $F_{p,q,n}$ as mapping the group element $g$ to the Gaussian transformation $s(g)$, which is well defined over the whole Fock space. 
\end{rem}



\section{de Finetti theorem}
\label{sec:dF}
 
The resolution of the identity for a family of generalized coherent states is a crucial ingredient for many of de Finetti theorems stating that tracing out a few subsystems of a state in the symmetric subspace gives a state that can be well approximated by a mixture of coherent states \cite{CFS02},\cite{KR05},\cite{CKMR07}.
In particular, K\"onig and Mitchison showed that such results follow from a general de Finetti theorem for representations of symmetry groups \cite{KM09}. In this sense, the Gaussian de Finetti theorem below can be seen as an application of the result of \cite{KM09} to the non compact group $SU(p,q)$ for which the de Finetti states are mixtures of Gaussian states. For completeness, we provide a proof of this result in this section.

The approximation level is expressed in terms of the trace distance between two operators, which is induced by the trace norm $\|A\|_{\tr} := \frac{1}{2} \tr |A|$ for any trace-class operator $A \in \mathfrak{B}(F_{p,q,n})$.

\begin{reptheorem}{thm:finetti}
Let $n$ be an arbitrary integer and $k\geq p+q$. Let $\rho = |\psi\rangle\langle \psi|$ be a symmetric (pure) state in $F_{p,q,n+k}^{U(n+k)}$. Then the state obtained after tracing out over $k(p+q)$ modes can be well approximated by a mixture of generalized coherent states: 
\begin{align*}
\left\|\tr_k  (\rho) - C_k \int \nu(\Lambda) |\Lambda,n \rangle \langle \Lambda, n| \d\mu_{p,q}(\Lambda) \right\|_{\tr}\leq \frac{3npq}{2(n+k-p-q)},
\end{align*}
with the density $\nu(\Lambda) := |\langle \Lambda, n+k| \psi\rangle|^2$ and $\d \mu_{p,q} := \frac{1}{C_n} \d \mu_{p,q,n}$.
\end{reptheorem}
The proof below follows closely that of Ref.~\cite{CKMR07},\cite{DOS07}.

\begin{proof} 

Tracing out $k$ modes of a state $\rho = |\psi\rangle\langle \psi|$ in the symmetric subspace $F_{n,p,q}^{U(n+k)}$ gives
\begin{align*}
\tr_k (\rho) = \tr_k  \left[\rho \left(\mathbbm{1} \otimes \1_{F_{p,q,k}^{U(k)}}\right) \right]
= C_k \int \rho_\Lambda \d \mu(\Lambda)
\end{align*}
where we used the resolution of the identity of Theorem \ref{thm:resol} in the second equality and defined $\rho_\Lambda := \tr_k \left[ (\1 \otimes \Pi_\Lambda^k) \rho \right]$, with $\Pi_{\Lambda}^k$ denoting the projector onto the state $|\Lambda,k\rangle$.

Projecting $\rho_\Lambda$ on $\Pi_{\Lambda}^n$ yields $\Pi_\Lambda^{n} \rho_\Lambda \Pi_{\Lambda}^{n} = \nu(\Lambda) \Pi_\Lambda^{n}$, with the normalization $\nu(\Lambda) := \tr[\rho_\Lambda \Pi_\Lambda^{n}]=\tr[\rho \Pi_\Lambda^{n+k}]$. 
We approximate $\tr_k (\rho)$ by the mixture of coherent states $C_k \tr_k \left[ \int \Pi_{\Lambda}^{n+k} \rho \Pi_{\Lambda}^{n+k} \d \mu(\Lambda)\right] = C_k \int \Pi_\Lambda^n \rho_\Lambda \Pi_{\Lambda}^n \d\mu(\Lambda)$. Here, we simply write $\d\mu$ instead of $\d\mu_{p,q}$ since $p$ and $q$ are fixed.
Applying the triangle inequality to $\rho - \Pi \rho \Pi = (\rho - \Pi \rho) + (\rho- \rho \Pi) - (\1- \Pi) \rho (1-\Pi)$ yields:
\begin{align*}
\left\|\tr_k  \rho - C_k \int \nu(\Lambda) \Pi_{\Lambda}^{n} \d\mu(\Lambda) \right\|_{\tr} = \left\| C_k \int (\rho_{\Lambda} - \Pi_\Lambda^{n} \rho_\Lambda \Pi_{\Lambda}^{n} ) \d\mu(\Lambda) \right\|_{\tr} \leq 2\Delta_1 + \Delta_2
\end{align*}
with $\Delta_1 :=  \left\| C_k \int \left(\rho_{\Lambda} - \Pi_\Lambda^{n} \rho_\Lambda \right) \d\mu(\Lambda) \right\|_{\tr}$ and $ \Delta_2 :=  \left\| C_k \int \left( \1- \Pi_\Lambda^{n}) \rho_\Lambda (\1-\Pi_{\Lambda}^{n}\right) \d\mu(\Lambda) \right\|_{\tr}$.

We bound both terms separately.
First, using that $\Pi_{\Lambda}^n \rho_\Lambda = \tr_k \left[ \Pi_{\Lambda}^{n+k} \rho \right]$ together with the resolution of the identity on $F_{p,q,n+k}^{U(n+k)}$, we obtain
\begin{align*}
\Delta_1 &= \left\| C_k \int \left(\rho_{\Lambda} - \tr_k \left[ \Pi_{\Lambda}^{n+k} \rho \right] \right) \d\mu(\Lambda) \right\|_{\tr} \\
&= \left\| \tr_k  (\rho)  - \frac{C_k}{C_{n+k}}   \tr_k \left[ \1_{F_{p,q,n+k}^{U(n+k)}  } \rho \right] \right\|_{\tr} \\
&= \left( 1 - \frac{C_k}{C_{n+k}}\right) \left\|  \tr_k (   \rho) \right\|_{\tr}=  \frac{1}{2}\left( 1 - \frac{C_k}{C_{n+k}}\right).
\end{align*}
Then, the convexity of the trace norm together with the fact that $\left(\1- \Pi_\Lambda^n) \rho_\Lambda (\1-\Pi_{\Lambda}^n\right)$ is a non negative operator give:
\begin{align*}
\Delta_2 & \leq  C_k \int   \left\|\left(\1- \Pi_\Lambda^n) \rho_\Lambda (\1-\Pi_{\Lambda}^n\right)  \right\|_{\tr}\d\mu(\Lambda)\\
&= \frac{C_k}{2}\tr \int  (\1- \Pi_\Lambda^n) \rho_\Lambda   \d\mu(\Lambda)= \Delta_1.
\end{align*}
This establishes the de Finetti approximation:
\begin{align*}
\left\|\tr_k \, (\rho) - C_k \int \Pi_\Lambda^n \rho_\Lambda \Pi_{\Lambda}^n \d\mu(\Lambda) \right\|_{\tr}\leq \frac{3}{2}\left( 1 - \frac{C_k}{C_{n+k}}\right).
\end{align*}
Finally, from the definition of $C_n$ given in Eq.~\eqref{eqn:cn}, we obtain
\begin{align*}
\frac{C_k}{C_{n+k}} &= \prod_{i=0}^{q-1} \frac{(k-q+i)!(n+k-p-q+i)!}{(k-p-q+i)!(n+k-q+i)!}  \\
&= \prod_{i=0}^{q-1} \prod_{j=1}^{p}  \frac{k-p-q+i+j}{n+k-p-q+i+j} \quad \text{(as can be seen by inspection)} \\
&\geq \left[ \frac{k-p-q+1}{n+k-p-q+1}\right]^{pq}\\
&\geq \left[ 1- \frac{n}{n+k-p-q+1}\right]^{pq}\\
&\geq  1- \frac{npq}{n+k-p-q+1}
\end{align*}
\end{proof}

\begin{rem}
The de Finetti approximation above does not correspond to a normalized state, but this is not usually needed for applications. Nevertheless, similarly as in Ref.~\cite{CKMR07}, it is possible to show that $\rho_k$ is approximated by the (normalized) density operator $C_{n+k} \int \nu(\Lambda) \Pi_{\Lambda}^n \d \mu(\Lambda)$, in which case the numerical factor $3/2$ in the trace distance needs to be replaced by 2, since
\begin{align*}
\left\|C_{k} \int \nu(\Lambda) \Pi_{\Lambda}^n \d \mu(\Lambda) - C_{n+k} \int \nu(\Lambda) \Pi_{\Lambda}^n \d \mu(\Lambda)\right\|_{\tr} &= \left(1- \frac{C_k}{C_{n+k}}\right)   \left\| C_{n+k} \int \nu(\Lambda) \Pi_{\Lambda}^n \d \mu(\Lambda)\right\|_{\tr}\\
& = \frac{1}{2} \left(1- \frac{C_k}{C_{n+k}}\right).
\end{align*}
\end{rem}

While the result of Theorem \ref{thm:finetti} fails to provide bounds that could be useful to analyze the security properties of some continuous-variable quantum cryptography protocols such as \cite{GG02} and \cite{WLB04}, we note that an exponential version of this theorem similar to that of Renner \cite{ren07} (see also \cite{KM09}) could be useful in this context and improve on current proof techniques such as \cite{RC09}. We will not explore this question in detail here since it would be more fruitful to exploit the resolution of the identity of Section \ref{sec:CS} in order to generalize the so-called de Finetti reduction of \cite{CKR09} and thereby obtain an improvement over the best currently available security analysis \cite{LGRC13}. This approach is explored elsewhere \cite{lev17}.


\section{First example: $p=q=1$}
\label{sec:11}

In this section, we consider in some detail our first example, namely the case where $p=q=1$.
We will follow closely Gazeau \cite{gaz09} and Perelomov \cite{per72}, even though our realization differs from theirs.
For $n \geq 2$, an integer, we consider the unitary irreducible representation $g = \left( \begin{smallmatrix} \alpha & \beta\\ \overline{\beta} & \overline{\alpha} \end{smallmatrix} \right) \mapsto T_g$ of $SU(1,1)$ on the symmetric space $F_{1,1,n}^{U(n)}$
\begin{align*}
SU(1,1) & \to \fS(F_{1,1,n}^{U(n)})\\
g & \mapsto  \left[ \psi_{\lambda,n}(Z) \mapsto (T_g \psi_{\lambda,n})(Z) :=  \psi_{\lambda_g,n}(Z) \right]
\end{align*}
with $\lambda_g := \frac{{\alpha} \lambda +\overline{ \beta}}{{\beta} \lambda + \overline{\alpha}}$ and $Z = z_1 z'_1 + \ldots + z_n z'_n$.

This countable set of representations constitutes the holomorphic discrete series of representations of $SU(1,1)$, which is locally isomorphic to $SO(2,1)$, the three-dimensional Lorentz group \cite{bar47},\cite{GN46},\cite{har47}.
We note that the series are usually parameterized by the parameter $n/2$ in the literature, which is either an integer or half an integer strictly larger than $1/2$.

In Section $5.2.1$ of \cite{per86}, Perelomov gives a general realization of any representation of the discrete series for $SU(1,1)$ with two kinds of bosonic operators. 
Our approach leads to a different realization where the generators $K_+, K_-$ and $K_0$ of the Lie algebra $\mathfrak{su}(1,1)$ are given by
\begin{align*}
&K_+ := Z  = \sum_{i=1}^n z_i z_i', \quad K_-  := \Delta =  \sum_{i=1}^n \frac{\partial}{\partial z_i} \frac{\partial}{\partial z_i'} \\
&K_0   := \frac{1}{2}\left(n + \sum_{i=1}^n z_i \frac{\partial}{\partial {z_i}} + z_i' \frac{\partial}{\partial {z'_i}} \right).
\end{align*}
It is straightforward to verify that these generators satisfy the commutation relations of the Lie algebra. 
\begin{lemma}
For $n \geq 1$, the generators defined above satisfy
\begin{align*}
[K_0, K_\pm] = \pm K_\pm, \quad [K_-, K_+]=2K_0.
\end{align*}
\end{lemma}
\begin{proof}
Lemma \ref{lemma-K+K-} in Appendix \ref{sec:lemmas} shows that
\begin{align*}
[K_-,K_+] = \hat{n}_A + \hat{n}_B + n = 2K_0
\end{align*}
for $\hat{n}_A :=  \sum_{i=1}^n z_i \partial_{z_i}$ and $\hat{n}_B := \sum_{i=1}^n z_i \partial_{z'_i}$.
Lemma \ref{lemma:nA-K+-} establishes that $\hat{n}_A K_- = K_-(\hat{n}_A-1)$ and that $K_+ \hat{n}_A = (\hat{n}_A-1)K_+$. The same results hold for $\hat{n}_B$. In other words, $[\hat{n}_A, K_+] = K_+$ and $[\hat{n}_A, K_-] = - K_-$.
We obtain:
\begin{align*}
[K_0, K_\pm] = \frac{1}{2} \left( [\hat{n}_A, K_\pm] + [\hat{n}_B, K_{\pm}] \right)= [\hat{n}_A, K_\pm]  = \pm K_\pm.
\end{align*}
\end{proof}

Let us now compute the Casimir operator $\hat{C}_2 := K_0^2 - \frac{1}{2} (K_+ K_- + K_- K_+)$ associated with this representation. 
Note that it commutes with the generators of the Lie algebra (see Lemma \ref{lemma:casimir}), which implies by Schur's lemma that it is a scalar for any irreducible representation. In order to compute its value, it is sufficient to see how it acts on the vacuum state. This gives:
\begin{align*}
\hat{C}_2 = \frac{n}{2} \left(\frac{n}{2}-1\right) \mathbbm{1}.
\end{align*}

We conclude this section with the explicit expression of coherent states for $SU(1,1)$ and give an orthonormal basis of $F_{1,1,n}^{U(n)}$. 

\begin{theorem}[Orthonormal basis of $F_{1,1,n}^{U(n)}$ and $SU(1,1)$ coherent states]
\label{thm:bon11}
For $n \geq 1$, the family $\left\{ \psi_{k,n}, k \in \N\right\}$ forms an orthonormal basis of $F_{1,1,n}^{U(n)}$ with
\begin{align*}
\psi_{k,n} :=  \left[\frac{(n-1)!}{(n+k-1)!k!}\right]^{1/2} Z^k.
\end{align*}
The $SU(1,1)$ coherent states are given by
\begin{align*} 
\psi_{\Lambda, n} = (1-|\Lambda|^2)^{n/2} e^{\Lambda Z} = (1-|\Lambda|^2)^{n/2} \sum_{k=0}^\infty \sqrt{\tbinom{n+k-1}{k}} \Lambda^k \psi_{k,n}.
\end{align*}
\end{theorem}
This is identical to Chapter 5 of \cite{per86}.


\section{Second example: $p=q=2$}
\label{sec:22}
The case $p=q=2$ occurs naturally in applications where one considers mixed states on $F_{1,1,n}$ invariant under the action of the unitary group. As established in Theorem \ref{theo:purification}, such density matrices can then be purified in $F_{2,2,n}^{U(n)}$

While $SU(1,1)$ coherent states have been studied extensively in the literature, this is hardly the case of the $SU(2,2)$ coherent states that we consider in the present section.
Of course, Perelomov already addressed the general case of $SU(p,q)$ coherent states for arbitrary $p$ and $q$ in Refs. \cite{per72}, \cite{per86}, but he did not provide an explicit orthonormal basis for the space spanned by these coherent states. 
We conjecture such an explicit basis in the following, but are unfortunately not able to provide a proof. Rather we could verify the correctness of these expressions for small values of the parameters and made a plausible guess for the general expression. Proving that these families are indeed orthonormal appears quite challenging. 

The reason for the added difficulty here compared to the easy case of $F_{1,1,n}^{U(n)}$ (Theorem \ref{thm:bon11}) lies in the fact that subspaces with a given total degree are degenerate in general. In the case of $F_{1,1,n}^{U(n)}$ on the contrary, two monomials with different degrees are necessarily orthogonal, which is not the case for $F_{2,2,n}^{U(n)}$: for instance, the states $Z_{1,1} Z_{2,2}$ and $Z_{1,2} Z_{2,1}$ are not orthogonal.

Let us define the combinatorial coefficient $a_{r,s}^{\ell, m}$ with 4 integer parameters $\ell, m, r, s \in \N$ as 
\begin{align*}
a_{r,s}^{\ell, m} := (a_r^n a_s^n -a_{r-1}^n a_{s-1}^n ) a_\ell^{n+r+s} a_{m}^{n+r+s}
\end{align*}
with $a_k^n := \tbinom{n+k-1}{k}$.

We now formulate our conjecture for an explicit orthonormal basis of $F_{2,2,n}^{U(n)}$. We could only check it for small values of $\ell, m, r, s$ and arbitrary $n >1$.

\begin{conj} \label{thm:bonSU22}
The set $\left\{\psi_{r,s}^{\ell,m} \: : \: \ell, m, r, s \in \N\right\}$ forms an orthonormal basis of the symmetric subspace $F_{2,2,n}^{U(n)}$, with
\begin{align} 
\psi_{r,s}^{\ell,m} & := \frac{1}{ \ell! r! s! m! \sqrt{a_{r,s}^{\ell, m}}}\sum_{i=0}^{\min (r,s)} (-1)^i \frac{\tbinom{r}{i} \tbinom{s}{i}}{ \tbinom{n+r+s-2}{i}}   Z_{11}^{\ell+ i} Z_{12}^{r-i} Z_{21}^{s-i}Z_{22}^{m+i}.  \label{identity}
 \end{align}
\end{conj}

We now define the coherent states for $SU(2,2)$ and express them in the orthonormal basis specified above. Assuming that Conjecture \ref{thm:bonSU22} holds, we can prove the following.
\begin{conj}[$SU(2,2)$ coherent states] \label{thm:SU22-CS}
The coherent state $\psi_{\Lambda,n} = \det (1-\Lambda \Lambda^\dagger)^{n/2} \det \exp (\Lambda^T Z)$ associated with the $2\times 2$ matrix $\Lambda$ is given by
\begin{align*}
\psi_{\Lambda,n} =  \det (1-\Lambda \Lambda^\dagger)^{n/2} \sum_{\ell,r,s,m=0}^\infty  \lambda_1^{\ell} \lambda_2^{r} \lambda_3^{s} \lambda_4^{m} \,   _2  F_1\left(-\ell,-m; n+r+s;  \frac{\lambda_2 \lambda_3}{\lambda_1 \lambda_4}\right) \sqrt{a_{r,s}^{\ell,m}} \psi_{r,s}^{\ell,m},
\end{align*}
where the hypergeometric function $_2 F_1$ is defined by
\begin{align*}
 \,   _2  F_1(-\ell,-m; n+r+s; \mu) := \sum_{i=0}^{\infty} \mu^{i}  \frac{\tbinom{\ell}{i} \tbinom{m}{i}}{ \tbinom{n+r+s+i-1}{i}}.
\end{align*}
\end{conj}

\begin{proof}
Let us expand $\exp(\lambda_1 Z + \lambda_2 V + \lambda_3 W +\lambda_4 T)$ and express it in the orthonormal basis of Theorem \ref{thm:bonSU22}:
\begin{align*}
\exp(\lambda_1 Z + \lambda_2 V + \lambda_3 W +\lambda_4 T) &= \sum_{\ell, r,s,m=0}^\infty \frac{\lambda_1^\ell \lambda_2^r \lambda_3^s \lambda_4^m}{\ell! r! s! m!} Z^\ell V^r W^s T^m \\
&= \sum_{\ell, r,s,m=0}^\infty \lambda_1^\ell \lambda_2^r \lambda_3^s \lambda_4^m  \sum_{i=0}^{\min(r,s)} \frac{\tbinom{\ell +i}{i} \tbinom{m+i}{i}}{ \tbinom{n+r+s-i-1}{i}}  \sqrt{a_{r-i,s-i}^{\ell+i,m+i}} \psi_{r-i,s-i}^{\ell+i, m+i} 
\end{align*}
The change of variables $\ell+i \to \ell, m+i \to m, r-i \to r, s-i \to s$ yields the result. 
\end{proof}


\section{An application: truncation of Haar random matrices}
\label{sec:marginal}

In their seminal work on Boson Sampling, Aaronson and Arkhipov raised the following question: how large should be $m$ in order for an $n \times n$ submatrix of an $m \times m$ Haar distributed unitary matrix to be close to a matrix of i.i.d. Gaussian entries, in total variation distance \cite{AA11}? 
The goal of this section is to show that the Gaussian de Finetti theorem of Section \ref{sec:dF} provides a nontrivial answer to this question by exploiting the $SU(n,n)$ coherent states. 

Following \cite{AA11}, let $m \geq n$ be two integers and define $\cU_{m,n}$ to be the Haar measure on the set of $m \times n$ complex matrices whose columns are orthogonal. We further define $\cH_{m,n}$ to be the distribution over $n \times n$ complex matrices obtained by first drawing a unitary $U$ from $\cU_{m,n}$ and then outputting $\sqrt{m} U_{n,n}$ where $U_{n,n}$ is the $n \times n$ submatrix of $U$ formed by its $n$ upper rows.
Let finally denote by $\cG^{n\times n}$ the probability distribution over $n \times n$ complex matrices whose entries are independent Gaussians with mean 0 and variance 1.

Aaronson and Arkhipov proved the following:
\begin{theorem}[\cite{AA11}]
Let $m \geq \frac{n^5}{\delta} \log^2 \frac{n}{\delta}$ for any $\delta>0$. Then $\|\cH_{m,n} - \cG^{n\times n}\|_{\mathrm{TV}} = O(\delta)$, where $\|\cdot\|_{\mathrm{TV}}$ denotes the total variation distance. 
\end{theorem}
Aaronson and Arkhipov further conjecture that their bound is not tight and that the right scaling should be $m = \Theta\left( \frac{n^2}{\delta} \right)$. This problem has applications in the context of Boson Sampling, which is a sampling problem that could potentially demonstrate the superiority of quantum processors compared to classical computers. There, $m$ and $n$ correspond respectively to the number of bosonic modes and the number of single photons required in a quantum implementation of the problem and reducing the number of modes in the implementation is of course of utmost importance for experimental realizations.

Apart from its application to Boson Sampling, similar questions have been studied in the literature for classical compact groups. See for instance Ref.~\cite{DF87} for a detailed history. 

Here we provide an improvement over the result of \cite{AA11} and show that:

\begin{reptheorem}{thm:trunc}
Let $m \geq n$. Then $\|\cH_{m,n} - \cG^{n\times n}\|_{\mathrm{TV}} \leq \frac{2n^3}{m-n}$. 
\end{reptheorem}
We note that the question of finding the correct scaling necessary for convergence for classical compact groups was recently settled by Jiang and Ma who showed that the distance goes asymptotically to 0 for $n=o(m^{1/2})$ \cite{JM17}. Our result is weaker in that respect, but follows almost directly from a application of the Gaussian de Finetti theorem and has the advantage of providing an explicit bound on the total variation distance.

Our proof strategy is the following: we first show that one can sample from a distribution arbitrarily close to $\cU_{m,n}$ by measuring a quantum state with heterodyne detection and from $\cH_{m,n}$ by measuring only $n^2$ modes of this state; next, exploiting our Gaussian de Finetti theorem, we prove that this marginal state is close to a mixture of Gaussian states, namely $SU(n,n)$ coherent states, which finally leads to our result provided the mixture is reduced to a single term, which holds in a suitable limit.

We introduce a series of quantum states $\rho_\alpha^1, \rho_\alpha^2, \rho_\alpha^3, \rho_\alpha^4$ depending on a real parameter $\alpha >0$ such that measuring $\rho_\alpha^4$ with heterodyne detection will allow us to sample from a distribution arbitrarily close to $\cH_{m,n}$ in the limit where $\alpha \to \infty$.
To this end, we associate to each of these states their respective (standard) Husimi $Q$ function, $Q_\alpha^1$ to $Q_\alpha^4$. 
In this section, we assume that $m$ and $n$ are fixed and omit to explicitly write these parameters in the name of the various states or distributions to simplify the notations. 

We start with the pure state $\rho_\alpha^1$ which is the $nm$-mode coherent state $|\alpha e_1\rangle |\alpha e_2\rangle \cdots |\alpha e_n \rangle$, where $\{e_1, \cdots, e_m\}$ is the canonical basis of $\C^m$ and $|\alpha e_i\rangle$ is the $m$-mode standard (Glauber) coherent state $|0, \cdots, 0, \alpha, 0, \cdots, 0\rangle$ with vacuum in all the modes except for the $i^\mathrm{th}$ one, and corresponding to the state $e^{-\alpha^2/2} e^{\alpha z_{i,i}}$.

The state $\rho_\alpha^2$ is obtained by twirling $\rho_\alpha^1$: 
\begin{align*}
\rho_{\alpha}^2 &:=  \int V_u \rho_\alpha^1 V_u^\dagger \d u
\end{align*}
where $\d u$ is the normalized Haar measure on $U(m)$ and $V_u |\alpha e_1\rangle |\alpha e_2\rangle \cdots |\alpha e_n \rangle = |\alpha u e_1\rangle |\alpha u e_2\rangle \cdots,$ $|\alpha u e_n\rangle$ for any unitary $u \in U(m)$.

The state $\rho_\alpha^2$ is mixed and invariant under $U(m)$ by construction. Using the recipe of Theorem \ref{theo:purification}, we define a purification that we denote by $\rho_\alpha^3 = |\Psi_\alpha^3\rangle \langle \Psi_\alpha^3|$, that belongs to the symmetric subspace $F_{n,n,m}^{U(m)}$.
This is a $2mn$-mode state. 

Finally, the state $\rho_{\alpha}^4$ is obtained by tracing out $(m-n)n + mn$ modes, which leaves a $n^2$-mode state in $\fS(F_{n,0,n})$: these $n^2$ modes are associated with variables $t_1, \cdots, t_n \in \C^n$.

This construction can be summarized as follows:
\begin{align*}
\rho_\alpha^1\xrightarrow{\text{twirling}} \rho_\alpha^2 \xrightarrow{\text{purif. (Thm \ref{theo:purification})}}  \rho_\alpha^3 \xrightarrow{\text{partial trace}} \rho_\alpha^4
\end{align*}

For a single-mode state $\rho$, the Husimi $Q$ function is simply defined as $Q(\alpha) = \frac{1}{\pi} \langle \alpha |\rho|\alpha\rangle$: it is the probability density function describing the measurement of the state with heterodyne detection. 
When dealing with multimode states, it will be convenient for us to group the variables of the $Q$ function as vectors of $\C^m$ or $\C^n$.  
For instance, the Husimi function $Q_\alpha^1$ associated with $\rho_\alpha^1$ is given by
\begin{align}
Q_\alpha^1(z_1, \cdots, z_n) = \frac{1}{\pi^{mn}} \exp\left(\sum_{i=1}^n |z_i -\alpha e_i|^2 \right),
\end{align}
where $z_1, \cdots, z_n \in \C^m$.

The Husimi function of $\rho_\alpha^2$ is obtained by twirling:
\begin{align}
Q_\alpha^2(z_1, \cdots, z_n) = \int Q_\alpha^1(u z_1, \cdots, u z_n) \d u = \frac{1}{\pi^{mn}} \int \exp\left(\sum_{i=1}^n |z_i -\alpha u e_i|^2 \right) \d u.
\end{align}
This function is related to the Haar distribution $\cU_{m,n}$ via a convolution with a Gaussian distribution.
Let us indeed define the complex Gaussian distribution 
\begin{align}
G_\alpha^{m\times n} (z_1, \cdots, z_n) := \left(\frac{\alpha^2}{\pi }\right)^{mn} \exp\left(- {\alpha^2}\sum_{i=1}^n |z_i|^2\right).
\end{align}
The following lemma is immediate.
\begin{lemma}
For $m,n \geq 1$ and $\alpha>0$, we have:
\begin{align}
\alpha^{2mn} Q_\alpha^2(\alpha z_1, \cdots, \alpha z_n) = \cU_{m,n} \star G_{\alpha}^{m \times n}(z_1, \cdots, z_n) 
\end{align}
for all $z_1, \cdots, z_n \in \C^m$.
\end{lemma}
As a consequence, it is possible to sample from a distribution arbitrarily close to $\cU_{m,n}$ by sampling from $Q_\alpha^2$ for $\alpha$ large enough and properly rescaling the output.
The Husimi function $Q_\alpha^3$ of the purification $\rho_\alpha^3$ is such that its marginal over $(z_1, \cdots, z_n)$ coincides with $Q_\alpha^2$:
\begin{align}
\int Q_\alpha^3(z_1, \cdots, z_n, z_1', \cdots, z_n') \d z_1' \cdots \d z_n' = Q_\alpha^2(z_1, \cdots, z_n)
\end{align}
and a similar result holds for $Q_\alpha^4$: denoting by $t_i \in \C^n$ the projection of $z_i$ onto the first $n$ coordinates, we obtain
\begin{align}
Q^4_\alpha(t_1, \cdots, t_n) = \int Q_\alpha^2(z_1, \cdots, z_n)  \prod_{i=1}^n\prod_{j=n+1}^m \d z_{i,j}.
\end{align}
The covariance matrix of the distribution $Q_\alpha^4$ is given by $\left(1+\frac{\alpha^2}{m}\right) \1_{n^2}$. Note that in general, the covariance matrix $\Gamma$ of (the Wigner function of) a state is related to the covariance matrix $\Gamma_Q$ of its Husimi $Q$ function through $\Gamma_Q = 1 + \frac{1}{2} \Gamma$. 
We finally define $Q_\alpha^5$, which is obtained by rescaling $Q_\alpha^4$ as follows:
\begin{align}
Q_\alpha^5(t_1, \cdots, t_n) := \left(1+\frac{\alpha^2}{m}\right)^{n^2/2} Q^4_\alpha\left(t_1 \sqrt{1+\frac{\alpha^2}{m}}, \cdots,  t_n \sqrt{1+\frac{\alpha^2}{m}}\right).
\end{align}
By construction of $Q_\alpha^5$ and by definition of the marginal $\cH_{m,n}$, we have the following:
\begin{lemma}
For $m,n \geq 1$ and $\alpha>0$,
\begin{align}
Q_\alpha^5(t_1, \cdots, t_n)= \cH_{m,n} \star G_{\sqrt{1+\frac{\alpha^2}{m}}}^{n\times n}(t_1, \cdots, t_n) 
\end{align}
for all $t_1, \cdots, t_n \in \C^n$. In particular, 
\begin{align} \label{eq:limit1}
\lim_{\alpha \to \infty} Q_\alpha^5 = \cH_{m,n}.
\end{align}
\end{lemma}

Our main tool to establish Theorem \ref{thm:trunc} is the Gaussian de Finetti theorem of Section \ref{sec:dF}. More precisely, applying this theorem to the pure state $\rho_\alpha^3$ in the symmetric subspace $F_{n,n,m}^{U(m)}$, that is for $p=q=n$, and tracing out the $B$ system (\ie variables $t_1', \cdots, t_n')$ yields:
\begin{align}\label{eq:8.10}
\left\|\rho_\alpha^4 -  \int f_\alpha(\Lambda) \tr_B (|\Lambda,n \rangle \langle \Lambda, n|) \d\Lambda \right\|_{\tr}\leq \frac{2n^3}{m-2n},
\end{align}
with the Lebesgue measure $\d \Lambda$ over $\mathcal{D}_{n,n}$ and 
\begin{align*}
f_\alpha(\Lambda) := C_m \frac{\langle \Lambda,m | \Psi_\alpha\rangle\langle \Psi_\alpha |\Lambda,m\rangle}{\det(\1_n - \Lambda \Lambda^\dagger)^{2n}}.
\end{align*}

One can easily relate the trace distance between two quantum states to the total variation distance between their Husimi functions. 
\begin{lemma}
Let $\rho_1$ and $\rho_2$ be two states, and let $Q_1$, $Q_2$ denote their respective $Q$-function. 
Then $\|Q_1 - Q_2 \|_{\mathrm{TV}} \leq \|\rho_1-\rho_2\|_{\mathrm{tr}}$.
\end{lemma}

\begin{proof}
The claim follows immediately from the operational interpretation of the trace distance as the advantage for the task of distinguishing $\rho_1$ and $\rho_2$. Since performing a heterodyne detection yields a probability density function given by the Husimi $Q$-function, we get that $\|Q_1 - Q_2 \|_{\mathrm{TV}}$ is an achievable advantage, which is therefore upper bounded by the trace distance between the states.
\end{proof}

Applying this lemma to Eq.~\eqref{eq:8.10} yields:
\begin{align}\label{eq:8.11}
\left\|Q_\alpha^4 -  \int f_\alpha(\Lambda) Q_\Lambda^{\otimes n} \d\Lambda \right\|_{\mathrm{TV}} \leq \frac{2n^3}{m-2n}.
\end{align}
The remaining step of the proof of Theorem \ref{thm:trunc} is to show that when properly rescaled, the distribution $\int f_\alpha(\Lambda) Q_\Lambda^{\otimes n} \d\Lambda$ tends to the Gaussian distribution $\cG^{n\times n}$ in the limit $\alpha \to \infty$.

We first establish some simple properties of $f_\alpha$.
\begin{lemma}\label{lem:invU}
The distribution $f_\alpha$ is invariant under conjugation by unitaries: for any unitary $U \in U(n)$, it holds that 
\begin{align}
f_\alpha(U \Lambda U^\dagger) = f_\alpha(\Lambda). 
\end{align}
\end{lemma}

\begin{proof}
Since $\det(\1_n - U\Lambda U^\dagger U \Lambda^\dagger U^\dagger) =\det(\1_n - \Lambda \Lambda^\dagger)$, it is sufficient to show that $|\langle U \Lambda U^\dagger, m| \Psi_\alpha^3 \rangle|^2 = |\langle  \Lambda, m| \Psi_\alpha^3 \rangle|^2$, where we write $\rho_\alpha^3 = |\Psi_\alpha^3\rangle\langle \Psi_\alpha^3|$.
Recall that $|\Psi_\alpha^3 \rangle = \lim_{\lambda \to 1} (1-\lambda^2)^{-mn/2} (\sqrt{\rho_\alpha^2} \otimes \1) |\lambda \1_n, m\rangle$.

Let us consider the action of $g_U = \left[\begin{smallmatrix} \overline{U} & 0 \\ 0 & \overline{U} \end{smallmatrix} \right] \in SU(n,n)$ for $U \in U(n)$, and $T_{g_U}$ its unitary representation on $F_{n,n,m}^{U(m)}$. 
It holds that $T_{g_U} |U \Lambda U^\dagger, m\rangle = |\Lambda,m\rangle$ and therefore
\begin{align*}
|\langle U \Lambda U^\dagger, m| \Psi_\alpha^3 \rangle|^2 &=|\langle U \Lambda U^\dagger, m| T_{g_U}^\dagger T_{g_U} |\Psi_\alpha^3 \rangle|^2\\
 &= \lim_{\lambda \to 1} (1-\lambda^2)^{-mn} |\langle  \Lambda, m|  T_{g_U} (\sqrt{\rho_\alpha^2} \otimes \1) |\lambda \1_n, m\rangle|^2\\
 &= \lim_{\lambda \to 1} (1-\lambda^2)^{-mn} |\langle  \Lambda, m|  T_{g_U} (\sqrt{\rho_\alpha^2} \otimes \1)T_{g_U}^\dagger T_{g_U} |\lambda \1_n, m\rangle|^2\\
 &= \lim_{\lambda \to 1} (1-\lambda^2)^{-mn} |\langle  \Lambda, m|  T_{g_U} (\sqrt{\rho_\alpha^2} \otimes \1)T_{g_U}^\dagger |\lambda \1_n, m\rangle|^2
 \end{align*}
where we used that $\lambda U\1_n U^\dagger = \lambda \1_n$ in the last equation.
It therefore only remains to show that the state $\rho_\alpha^2$ is invariant under the action of the unitary $U$ to conclude.

Let $z$ be the $m \times n$ matrix $[z_{ij}]$ of creation operators and $z^*$ be the matrix $[z^*_{ij}]$ of annihilation operators. Let us further introduce the $n \times m$ matrix $J$, with $J_{i,j}= \delta_{i,j}$. 
The state $\rho_{\alpha}^2$ is given by:
\begin{align*}
\rho_\alpha^2 &= \int \d V \exp(-n \alpha^2) \exp (\alpha \mathrm{tr} (JV z) ) \exp(\alpha \mathrm{tr} (J\overline{V} z^*) 
\end{align*}
where $\d V$ is the normalized Haar measure on $U(m)$.
Applying $T_{g_U}$ maps $z$ to $zU$ and $z^*$ to $z^* U^\dagger$ and
\begin{align*}
T_{g_U} \rho_\alpha^2 T_{g_U}^\dagger &= \int \d V \exp(-n \alpha^2) \exp (\alpha \mathrm{tr} (JV z U) ) \exp(\alpha \mathrm{tr} (J\overline{V} z^* \overline{U}))\\
&= \int \d v \exp(-n \alpha^2) \exp (\alpha \mathrm{tr} (UJV z) ) \exp(\alpha \mathrm{tr} (\overline{U} J\overline{V} z^* ))\\
&= \rho_\alpha^2
\end{align*}
where we used that the Haar measure is invariant by multiplication by a unitary matrix in the last equality.
This establishes the claim.
\end{proof}

For $\vec{k} = (k_1, \cdots, k_n) \in \N^n$, let $\Pi_{\vec{k}}$ be the projector onto the subspace of $F_{n,n,m}$ spanned by states containing $k_i$ photons in modes $A_{1,i}$ to $A_{m,i}$ for all $i \in [n]$. 
Let us further define the probability distributions $p^\alpha$ and $q^\Lambda$ corresponding to the outcome distributions of measuring $|\Psi_\alpha^3\rangle$ and $|\Lambda,m\rangle$ with the measurement $\{\Pi_{\vec{k}}\}_{\vec{k} \in \N^{n}}$.

\begin{lemma}\label{lem:fidel}
\begin{align}
f_\alpha(\Lambda) \leq  \frac{C_m}{\det(\1_n - \Lambda \Lambda^\dagger)^{2n}}  F(p^\alpha, q^\Lambda),
\end{align}
where $F(p,q) := \left( \sum_{i} \sqrt{p_i q_i} \right)^2$ denotes the fidelity between the distributions $p$ and $q$.
\end{lemma}

\begin{proof}
Since $\sum_{\vec{k}} \Pi_{\vec{k}} = \1_{F_{n,n,m}^{U(m)}}$, we obtain that 
\begin{align*}
f_\alpha(\Lambda) =  \frac{C_m}{\det(\1_n - \Lambda \Lambda^\dagger)^{2n}} \sum_{\vec{k}, \vec{\ell}} \langle \Lambda,m | \Pi_{\vec{k}} | \Psi_\alpha\rangle\langle \Psi_\alpha | \Pi_{\vec{\ell}} |\Lambda,m\rangle.
\end{align*}

Let us denote $\Pi_{\vec{k}} |\Psi_\alpha\rangle = \sum_{\vec{k}} \sqrt{p^\alpha_{\vec{k}}} |\phi_{\vec{k}}^\alpha\rangle$ and similarly $\Pi_{\vec{k}}|\Lambda,m\rangle =  \sum_{\vec{k}} \sqrt{q^\Lambda_{\vec{k}}} |\tilde{\phi}_{\vec{k}}^\Lambda\rangle$ where $|\phi_{\vec{k}}^\alpha\rangle$ and $|\tilde{\phi}_{\vec{k}}^\Lambda\rangle$ are normalized.
We note that the distributions $p^\alpha = (p^\alpha_{\vec{k}})_{\vec{k}}$ and $q^{\Lambda} = (q_{\vec{k}}^\Lambda)$ are normalized probability distributions over $\N^n$.
This allows us to write
\begin{align*}
f_\alpha(\Lambda) = |f_\alpha(\Lambda)|&=  \frac{C_m}{\det(\1_n - \Lambda \Lambda^\dagger)^{2n}} \left| \sum_{\vec{k}, \vec{\ell}} \sqrt{p^\alpha_{\vec{k}} p^\alpha_{\vec{\ell}} q^\Lambda_{\vec{k}} q^\Lambda_{\vec{\ell}}}  \langle \tilde{\phi}_{\vec{k}}^\Lambda | \phi_{\vec{k}}^{\alpha}Ê\rangle\langle\phi_{\vec{\ell}}^{\alpha}Ê |  \tilde{\phi}_{\vec{\ell}}^\Lambda \rangle \right|\\
& \leq   \frac{C_m}{\det(\1_n - \Lambda \Lambda^\dagger)^{2n}} \sum_{\vec{k}, \vec{\ell}} \sqrt{p^\alpha_{\vec{k}} p^\alpha_{\vec{\ell}} q^\Lambda_{\vec{k}} q^\Lambda_{\vec{\ell}}}\\
&=  \frac{C_m}{\det(\1_n - \Lambda \Lambda^\dagger)^{2n}} \left(\sum_{\vec{k}} \sqrt{p^\alpha_{\vec{k}}  q^\Lambda_{\vec{k}}} \right)^2 \\
&=  \frac{C_m}{\det(\1_n - \Lambda \Lambda^\dagger)^{2n}}  F(p^\alpha, q^\Lambda),
\end{align*}
where the inequality results from $\left| \langle \tilde{\phi}_{\vec{k}}^\Lambda | \phi_{\vec{k}}^{\alpha}Ê\rangle\langle\phi_{\vec{\ell}}^{\alpha}Ê |  \tilde{\phi}_{\vec{\ell}}^\Lambda \rangle \right| \leq 1$.
\end{proof}

We now show that $q^\Lambda$ and $Q_\Lambda^{\otimes n}$ are invariant when $\Lambda$ is replaced by $\Lambda V^\dagger$ for $V \in U(n)$. 
\begin{lemma}\label{lem:invV}
For any $V \in U(n)$, any $\Lambda \in \mathcal{D}_{n,n}$ and integer $n \geq 1$, it holds that
\begin{align*}
Q_{\Lambda V^\dagger}^{\otimes n} = Q_\Lambda^{\otimes n} \quad \text{and} \quad q^{\Lambda V^\dagger} = q^\Lambda.
\end{align*}
\end{lemma}

\begin{proof}
Since the two states are Gaussian, it is sufficient to show that their first two moments coincide. Their first moment is null and according to Theorem \ref{lem:CM}, their covariance matrices are given respectively by:
\begin{align*}
\cS((\1_n + \Lambda V^\dagger (\Lambda V^\dagger)^\dagger) (\1_n- \Lambda V^\dagger (\Lambda V^\dagger)^\dagger)^{-1}) \quad \text{and} \quad
\cS((\1_n + \Lambda\Lambda^\dagger) (\1_n- \Lambda \Lambda^\dagger)^{-1}) 
\end{align*}
which are equal by unitarity of $V$.
Since the states are equal, their photon number distributions also coincide.
\end{proof}

From Lemmas \ref{lem:invU}, \ref{lem:fidel} and \ref{lem:invV}, we deduce that for an arbitrary matrix $\Lambda = U \Sigma V^\dagger \in \mathcal{D}_{n,n}$, it holds that
\begin{align}
f_\alpha(\Lambda) &= f_\alpha( \Sigma V^\dagger U) \nonumber\\
&\leq  \frac{C_m}{\det(\1_n - \Sigma^2)^{2n}}  F(p^\alpha, q^{\Sigma V^\dagger U}) \nonumber\\
&= \frac{C_m}{\det(\1_n - \Sigma^2)^{2n}}  F(p^\alpha, q^{\Sigma }). \label{eq:final-ineq}
\end{align}
It is easy to determine the distribution $p^\alpha$: it is simply the product of $n$ Poisson distributions of parameter $\alpha^2$. For $\alpha$ large enough, it is well approximated by $n$ Gaussian variables with mean $\alpha^2$ and variance $\alpha^2$.

Similarly, the distribution $q^\Sigma$ is the $n$-fold convolution of the photon number distribution $q_1^\Sigma$ of the $n$-mode thermal state corresponding to half of the state $|\Sigma,1\rangle$ for $\Sigma = \mathrm{diag}(\Sigma_1, \cdots, \Sigma_n)$.
The distribution $q^\Sigma$ is the $m$-fold convolution of
\begin{align*}
q_1^\Sigma(k_1, \cdots, k_n) = \prod_{i=1}^n (1-\Sigma_i^2) \Sigma_i^{2k_i}.
\end{align*}
This is a product of geometric distributions and we deduce that $q^\Sigma$ is a product negative binomial distributions:
\begin{align*}
q^\Sigma(k_1, \cdots, k_n) = \prod_{i=1}^n (1-\Sigma_i^2)^m \tbinom{m+k_i-1}{k_i} \Sigma_i^{2k_i}.
\end{align*}
The mean and variance of $q^\Sigma$ are given by
\begin{align*}
\mathbbm{E} = \left(\frac{m \Sigma_1^2}{1-\Sigma_1^2}, \cdots, \frac{m \Sigma_n^2}{1-\Sigma_n^2} \right), \quad \mathrm{Var} = \mathrm{diag} \left(\frac{m \Sigma_1^2}{(1-\Sigma_1^2)^2}, \cdots, \frac{m \Sigma_n^2}{(1-\Sigma_n^2)^2} \right).
\end{align*}

From Eq.~\eqref{eq:final-ineq}, we conclude that the only matrices $\Lambda$ for which $f_\alpha(\Lambda)$ is nonnegligible in the limit $\alpha \to \infty$ are those for which the fidelity between $p^\alpha$ and $q^\Sigma$ is large. This only holds when both distributions have approximately the same mean. This statement becomes exact in the limit $\alpha \to \infty$ since the distribution $p^\alpha$ becomes more concentrated around its mean value. 
In particular, this imposes that 
\begin{align*}
\Sigma_1^2 \approx \cdots \approx \Sigma_n^2 \approx \frac{\alpha^2}{m + \alpha^2}.
\end{align*}
This means that the only matrices $\Lambda$ contributing in a nonnegligible fashion to the distribution $f_\alpha$ in the limit $\alpha \to \infty$ are close to multiples of unitaries: $\Lambda \approx \lambda U$ for $U \in U(n)$ and $\lambda = \sqrt{\frac{\alpha^2}{m + \alpha^2}}$.
Furthermore, according to Lemma \ref{lem:invV}, for any $U \in U(n)$, it holds that $Q_{\lambda U} = Q_{\lambda \1_n}$. Note finally that $Q_{\lambda \1_n}$ is the probability distribution over $n\times n$ complex matrices whose entries are independent Gaussians with mean 0 and variance $1+ \frac{1+\lambda^2}{2(1-\lambda^2)} = \frac{3}{2} + \frac{\alpha^2}{m}$.
This establishes the following statement about the rescaled version of $\int f_\alpha(\Lambda) Q_\Lambda^{\otimes n} \d \Lambda$
\begin{lemma}
Define $Q_{\Lambda, \alpha}(t_1, \cdots, t_n) := \left(1+ \frac{\alpha^2}{m}\right)^{n/2} Q_\Lambda\left(  t_1\sqrt{1+\frac{\alpha^2}{m}}, \cdots,  t_n \sqrt{1+\frac{\alpha^2}{m}}\right)$, then
\begin{align}\label{eq:limit2}
\lim_{\alpha \to \infty} \int f_\alpha (\Lambda) Q^{\otimes n}_{\Lambda, \alpha} \d \Lambda = \cG^{n\times n}.
\end{align}
\end{lemma}
Performing the change of variable $(t_1, \cdots, t_n) \to \left(t_1\sqrt{1+\frac{\alpha^2}{m}}, \cdots, t_n\sqrt{1+\frac{\alpha^2}{m}} \right)$ in Eq.~\eqref{eq:8.11} yields
\begin{align}
\left\|Q_\alpha^5 -  \int f_\alpha (\Lambda) Q_{\Lambda, \alpha}^{\otimes n}\d \Lambda \right\|_{\mathrm{TV}} \leq \frac{2n^3}{m-2n},
\end{align}
The proof of Theorem \ref{thm:trunc} follows from Eq.~\eqref{eq:limit1} and \eqref{eq:limit2}.


\newpage

\appendix

\section{Characterization of the symmetric subspace: proof of Theorem \ref{thm:charact-symm}}
\label{sec:charac}

In order to prove Theorem \ref{thm:charact-symm}, we are essentially interested in understanding the space of complex entire functions in the $n(p+q)$ variables $z_{1,1}, \ldots, z_{n,p}, z'_{1,1}, \ldots, z'_{n,q}$ which are invariant under the change of variables $z_{i} \to u z_{i}, z'_{j} \to \overline{u}z'_j$ for all unitary $u\in U(n)$, where $z_i = (z_{1,i}, \ldots, z_{n,i})^T$ and $z'_j = (z'_{1,j}, \ldots, z'_{n,j})$. 
Such a function of $n(p+q)$ variables can be concisely written as $\sum_{A,B} [A|B] z^A z'^B$ where $A$ and $B$ are respectively an $n \times p$ and an $n \times q$ tables of integers $A_{t,i}, B_{u,j} \in \N$ for $t, u \in [n], i \in [p]$ and $j\in [q]$, $[A|B]$ is a complex-valued coefficient and $z^A$ and $z'^{B}$ are defined as
\begin{align*}
z^A := \prod_{t \in [n]}\prod_{i \in [p]} z_{t,i}^{A_{t,i}} \quad \text{and} \quad z'^B :=\prod_{u \in [n]}\prod_{j \in [q]} z_{u,j}'^{B_{u,j}}.
\end{align*}

In this section, it will be useful to introduce a notation in order to refer to a specific entry of the matrix $A$ or $B$ of a coefficient $[A|B]$: in the expressions $\alpha_{t,i} [A|B]$ and $\beta_{u,j}[A|B]$, the integer $\alpha_{t,i}$ corresponds to $A_{t,i}$ and $\beta_{u,j}$ to $B_{u,j}$. 
This notation is convenient when the matrices in the bracket are sums of matrices. 
We denote by $e_{t,i}$ (resp.~$e_{u,j}$) the $n \times p$ (resp.~$n \times q$) matrix with a one at coordinates $(t,i)$ (resp.~$(u,j)$) and 0 elsewhere. (It should always be clear from context whether $e_{t,i}$ is an $n \times p$ or an $n\times q$ matrix.)
In that case, our convention tells us that 
\begin{align*}
\alpha_{t,i} [A  -e_{u, i} + e_{t,i} | B] &= (A_{t,i} -\delta_{t,u} + 1) [A  -e_{u, i} + e_{t,i} | B]\\
\beta_{u,j} [A  | B +e_{u, j} -e_{t,j} ]&= (B_{u,j}+1-\delta_{u,t}) [A  | B +e_{u, j} -e_{t,j} ]
\end{align*}
where $\delta_{i,j}$ is the Kronecker symbol equal to 1 if $i=j$ and 0 otherwise. 

Finally, we are not interested here in convergence properties of the entire functions and will therefore consider the space $\C[[z_{1,1}, \ldots, z_{n,p}, z'_{1,1}, \ldots, z'_{n,q}]]$ of formal series in $z_{1,1}, \ldots, z_{n,p}, z'_{1,1}, \ldots, z'_{n,q}$. 
For brevity, let us denote by $\C[[z,z']]$ the ring of formal power series $\C[[z_{1,1}, \ldots, z_{n,p}, z'_{1,1}, \ldots, z'_{n,q}]]$ in the $n(p+q)$ variables $z_{1,1}, \ldots, z_{n,p}, z'_{1,1}, \ldots, z'_{n,q}$.
As before, we define the finite-dimensional restrictions $E_{p,q,n}^{\leq d}$ (resp.~$F_{p,q,n}^{U(n), \leq d}$) of these spaces consisting of polynomials of degree at most $d$ in the variables $Z_{1,1}, \ldots, Z_{p,q}$ (resp.~$z_{1,1},\ldots, z_{n,p}$).
Our goal in this section is to show that $E_{p,q,n}^{\leq d} = F_{p,q,n}^{\leq d, U(n)}$ and we already know by Lemma \ref{lem:Z-inv} that $E_{p,q,n}^{\leq d} \subseteq F_{p,q,n}^{\leq d, U(n)}$.

We have the following characterization of $F_{p,q,n}^{U(n)}$.
\begin{theorem}\label{thm:charac1}
An entire function $\sum_{A,B} [A|B] z^A z'^B$ is invariant under the change of variables $z_{i} \to u z_{i}, z'_{j} \to \overline{u}z'_j$ for all unitaries $u\in U(n)$  if and only if the coefficients $[A|B] \in \mathbbm{C}$ safisfy the equations
\begin{align*}
(E^{t,u}) \quad \quad \sum_{i=1}^p \alpha_{t,i} [A  -e_{u, i} + e_{t,i} | B] = \sum_{j=1}^q \beta_{u,j} [A  | B +e_{u, j} -e_{t,j} ]
\end{align*}
for all $t, u \in [n]$, where the notation $\alpha_{t,i}$ refers to the entry $(t,i)$ of the matrix $A  -e_{u, i} + e_{t,i}$ and $\beta_{u,j}$ refers to the entry $(u,j)$ of the matrix $B +e_{u, j} -e_{t,j}$.
\end{theorem}

\begin{proof}
An entire function is invariant under the action of $U(n)$ if and only if it is invariant under the action of $U_{t,\theta} := e^{i\theta} e_{t,t}$ for all $ t\in [n], \theta \in [0,2\pi[$ and $U_{t, u,\theta} = \cos \theta (e_{t,t}+ e_{u, u}) + \sin \theta (e_{t,u} - e_{u,t})$, for all $t\ne u$ and $\theta \in [0, 2\pi[$. Here, $e_{t,u}$ is an $n\times n$ matrix with a one at coordinates $(t,u)$ and 0 elswhere.
This is equivalent to being invariant under these actions in the limit $\theta \to 0$, that is, being invariant under the action of the generators of the Lie algebra $\mathfrak{u}(n)$, \ie the real vector space of $n\times n$ skew-Hermitian matrices.

The change of variables associated with $U_{t, \theta}$ maps $(z_{t,i}, z'_{t,j})$ to $(e^{i\theta}  z_{t,i}, e^{-i\theta} z'_{t,j})$ and leaves $(z_{u, i}, z'_{u, j})$ unchanged for $u \ne t$. 
The function $\sum_{A,B} [A|B] z^A z'^B$ is then mapped to $\sum_{A,B} \exp(i \theta (\sum_{j} a_{t,j} -\sum_{j} b_{t,j} ) ) [A|B] z^A z'^B$. 
In particular, the function is invariant under the action of all unitaries of the form $U_{t, \theta}$ if and only if 
\begin{align}
\left(\sum_{j} a_{t,j} -\sum_{j} b_{t,j} \right) [A|B] = 0 \label{eqn:D0proof} 
\end{align}
for all $A, B$. This corresponds to all the equations of the form $(E^{t,t})$.

Let us now consider the unitary transformation $U_{t,u,\theta}$. The entire series $P$ is then mapped to $P_{t,u}(\theta)$. Taking the limit $\theta \to 0$ and imposing that $P = P_{t,u} $ gives $\frac{\partial P_{t,u}}{\partial \theta} = 0$ when evaluating in $\theta=0$. 
Differentiating with respect to $\theta$ gives:
\begin{align*}
\frac{\partial P_{t,u}}{\partial \theta} = \sum_{i=1}^p \frac{\partial P_{t,u }}{\partial z_{t,i}} \frac{{\partial z_{t,i}}}{\partial \theta} +\frac{\partial P_{t,u }}{\partial z_{u,i}} \frac{{\partial z_{u,i}}}{\partial \theta} +\sum_{j=1}^q \frac{\partial P_{t,u }}{\partial z'_{t,j}} \frac{{\partial z'_{t,j}}}{\partial \theta} +\frac{\partial P_{t,\ell }}{\partial z'_{u,j}} \frac{{\partial z'_{u,j}}}{\partial \theta}.
\end{align*}
Under the unitary $U_{t,u}$, one has
\begin{align*}
\left. \frac{{\partial z_{t,i}}}{\partial \theta}\right|_{\theta=0} = z_{ u,i}, \quad \left. \frac{{\partial z_{u,i}}}{\partial \theta}\right|_{\theta=0} =- z_{ t,i}, \quad
\left. \frac{{\partial z'_{t,j}}}{\partial \theta}\right|_{\theta=0} = z'_{ u,j}, \quad \left. \frac{{\partial z'_{u,j}}}{\partial \theta}\right|_{\theta=0} = -z'_{ t,j}.
\end{align*}
Injecting this in the previous expression leads to
\begin{align*}
0 = \left. \frac{\partial P_{t,u}}{\partial \theta}\right|_{\theta=0} = \sum_{i=1}^p z_{u,i} \frac{\partial P_{t,u }}{\partial z_{t,i}}  - z_{t,i} \frac{\partial P_{t,u }}{\partial z_{u,i}}+\sum_{j=1}^q -z_{u,j} \frac{\partial P_{t,u }}{\partial z'_{t,j}}  +z'_{ t,j} \frac{\partial P_{t,u }}{\partial z'_{u,j}}.
\end{align*}
Applying this operator to $\sum_{A,B} [A|B] z^A z'^B$ and asking that the monomial $z^A z'^B$ has a zero coefficients yields
\begin{align*}
0 = \sum_{i=1}^p \alpha_{t,i} [A-e_{u,i} + e_{t,i}|B]   - \alpha_{u,i} [A  + e_{u,i} - e_{t,i} |B] +\sum_{j=1}^q  \beta_{t,j} [A |ÊB -e_{u,j} + e_{t,j} ] - \beta_{u,j} [A | B -e_{t,j}  + e_{u,j}]
\end{align*}
One should finally note that two of the four terms above are necessarily null, because of Eq.~\eqref{eqn:D0proof}.
In particular, assume that $A, B$ are such that $[A-e_{u, i}+ e_{t,i}|B] \ne 0$, then Eq.~\eqref{eqn:D0proof} implies that $[A  + e_{u,i} - e_{t,i} |B] = [A |ÊB -e_{u,j} + e_{t,j} ] =0$, and the equation becomes
\begin{align*}
\sum_{i=1}^p \alpha_{t,i} [A-e_{u,i} + e_{t,i}|B]  \beta_{u,j} [A | B -e_{t,j}  + e_{u,j}] =0. 
\end{align*}
This gives the set of equations $(E^{t,u})$ for $t\ne u$, as was to be proved. One can check that choosing $A, B$ such that $[A  + e_{u,i} - e_{t,i} |B] \ne 0$ leads to the same set of equations.
\end{proof}

Armed with this characterization of Theorem \ref{thm:charac1}, the next step is to study which formal power series satisfy all equations $(E^{t,u})$ for $t, u \in [n]$.

\begin{defn}
For indices $i \in [p]$ and $j \in [q]$, the \emph{Laplacian operator} $\Delta$ is the linear map 
\begin{align*}
\C[[z,z']] & \to \C[[z,z']]\\
P &\mapsto \Delta_{i,j} P := \sum_{t=1}^n \frac{\partial}{\partial{z_{t,i}}} \frac{\partial}{\partial{z'_{t,j}}} P.
\end{align*}
For a $p \times q$ table $M = (m_{i,j}) \in \N^{pq}$, we define the application $\Delta^M := \prod_{i=1}^p \prod_{j=1}^q \Delta_{i,j}^{m_{i,j}}$. 
Finally, for an integer $k \geq 1$, the application $\Delta^k$ is defined as
\begin{align*}
\C[[z,z']] &\to \C[[z,z']]^K\\
P & \mapsto (\Delta^M P)_{\{|M|=k\}}
\end{align*}
where $K={\tbinom{pq+k-1}{k}}$ is the number of tables $M \in \N^{pq}$ with weight $k$.
\end{defn}

\begin{lemma}
For any $(i,j) \in [p]\times [q]$, the operator $\Delta_{i,j}$ is an endomorphism of $F_{p,q,n}^{U(n)}$.
\end{lemma}
\begin{proof}
The proof is similar to that of Lemma \ref{lem:Z-inv} since it suffices to show that the operator $\Delta_{k,\ell}$ is invariant under the action of $U(n)$. 
\end{proof}

\begin{theorem} \label{thm:trivial-kernel} For $p,q \geq 1, n\geq 1$, the kernel of the Laplacian operator $\Delta$ in $F_{p,q,n}^{U(n)}$ is trivial:
\begin{align*}
\ker (\Delta) \cap F_{p,q,n}^{U(n)} = \bigcap_{(i,j) \in [p]\times [q]} \ker (\Delta_{i,j}) \cap F_{p,q,n}^{U(n)} = \mathbbm{C}.
\end{align*}
\end{theorem}

The theorem is proved in Appendix \ref{sec:proof21}.

This allows us to upper bound the dimension of $F_{p,q,n}^{U(n), \leq d}$, for $d \geq0$. Recall that for integers $k,n \geq 0$, we define $a_k^n := \tbinom{k+n-1}{k}$, which is the dimension of the symmetric tensor $S H^{\otimes k}$ if the Hilbert space $H$ has dimension $n$.
\begin{lemma}\label{lem:dimW}
The dimension of $F_{p,q,n}^{U(n), \leq d}$ satisfies:
\begin{align*}
\mathrm{dim} \,F_{p,q,n}^{U(n), \leq d} \leq  \sum_{k=0}^d  \tbinom{pq+k-1}{k}= \tbinom{pq+d}{d}.
\end{align*}
\end{lemma}

\begin{proof}
Let $d \geq 0$ be a fixed degree.
Let us introduce the space $V_k = F_{p,q,n}^{U(n), \leq d} \cap \ker \Delta^k$. Since $F_{p,q,n}^{U(n), \leq d}$ is restricted to polynomials of degree at most $d$ in $z$, we have that $F_{p,q,n}^{U(n), \leq d} \subseteq \ker \Delta^{d+1}$, and therefore $V_{d+1} =F_{p,q,n}^{U(n), \leq d}$.
Our goal is therefore to upper bound the dimension of $V_{d+1}$. 

Let us proceed by induction with the induction hypothesis $(H_k)_{k \geq 1}$ given by
\begin{align*}
(H_k) \quad \mathrm{dim} \, V_{k} \leq \sum_{i=0}^{k-1} \tbinom{pq+i-1}{i}.
\end{align*}
The base case $k=1$ corresponds to Theorem \ref{thm:trivial-kernel}: $\mathrm{dim} \ker \Delta = 1$, which gives $\mathrm{dim} \, V_1 = 1$.

Assume now that $H_k$ holds. 
The rank-nullity theorem applied to $\Delta^{k}$ restricted to the space $V_{k+1}$ gives
\begin{align*}
\mathrm{dim} \, V_{k+1} &= \mathrm{dim}\, \left(V_{k+1} \cap \ker \Delta^{k}\right) + \mathrm{dim} \, \mathrm{Im} \Delta^{k}.
\end{align*}
It holds that
\begin{align*}
V_{k+1} \cap \ker \Delta^{k} = F_{p,q,n}^{U(n), \leq d} \cap \ker \Delta^{k+1} \cap \ker \Delta^{k}  = F_{p,q,n}^{U(n), \leq d} \cap \ker \Delta^{k}  = V_k
\end{align*}
since $\ker \Delta^{k} \subseteq \ker \Delta^{k+1}$.
The induction hypothesis gives:
\begin{align*}
\mathrm{dim} \, V_{k+1} \leq \sum_{i=0}^{k-1} \tbinom{pq+i-1}{i} + \mathrm{dim} \, \mathrm{Im} \Delta^{k-1}.
\end{align*}
Consider a element $Q \in  \mathrm{Im} \Delta^{k}$. 
It can be written as $Q = \Delta^{k} P$ for some polynomial $P \in \ker \Delta^{k+1}$ and consists of $\tbinom{pq +k-1}{k}$ polynomials corresponding to the various $\Delta^{M}P$ for all tables $M \in \N^{pq}$ with $|M|=k$.
Since $P \in \ker \Delta^{k+1}$, it holds that for each such table $M$, $\Delta (\Delta^{M}) P = 0$. Theorem \ref{thm:trivial-kernel} shows that $\Delta^{M} \in \C$. This means that $\mathrm{dim} \, \mathrm{Im} \Delta^{k-1} \leq \tbinom{pq +k-1}{k}$, which proves the induction step. 
\end{proof}

\begin{lemma}\label{lem:acts-freely}
For $n \geq \min(p,q)$, 
\begin{align*}
\mathrm{dim}\, E_{p,q,n}^{\leq d}  = \sum_{k=0}^d \tbinom{pq+k-1}{k}.
\end{align*}
\end{lemma}
\begin{proof}
Without loss of generality, let us assume that $\min(p,q)=p$. 
Let us consider the following partition of the set $[p] \times [q]$:
\begin{align*}
[p] \times [q] = \bigcup_{k=1}^{p} S_k
\end{align*}
with $S_k := \left\{(k,\ell) \: : \: k \leq \ell \leq q Ê\right\} \cup \left\{(\ell,k) \: : \: k \leq \ell \leq p Ê\right\}$, for $k \in [p]$.
This partition allows us to decompose in a unique fashion any table $M \in \N^{p,q}$ as:
\begin{align*}
M = \Pi_1 (M) + \ldots + \Pi_p (M)
\end{align*}
where $\Pi_k(M)$ is the table with entries coinciding with those $M$ for the coordinates in $S_k$ and with entries equal to 0 elsewhere. 
In other words,
\begin{displaymath}
[\Pi_k M]_{i,j} = \left\{
\begin{array}{ll}
M_{i,j} & \text{if} \, k = \min(i,j),\\
0 & \text{else.} 
\end{array} \right.
\end{displaymath}
For any table $M = (m_{i,j}) \in \N^{pq}$, the polynomial $Z^M := \prod_{ i \in [p] }\prod_{j \in [q]} Z_{i,j}^{m_{i,j}}$ factorizes as
\begin{align*}
Z^M  = \prod_{k=1}^{p}  Z^{\Pi_k(M)}= \prod_{k=1}^{p} \left( \prod_{(i,j) \in S_k } Z_{i,j}^{\left[\Pi_k(M)\right]_{i,j}}\right).
\end{align*}
We now argue that the monomial
\begin{align*}
P_M(z_{1,1}, \ldots, z_{n,p}, z'_{1,1}, \ldots, z'_{n,q}) &:= \prod_{k=1}^{p} \left( \prod_{(i,j) \in S_k } (z_{k,i} z'_{k,j})^{\left[\Pi_k(M)\right]_{i,j}}\right) 
\end{align*}
only appears in the decomposition of $Z^M$, and not in any other $Z^{N}$ for a table $N \ne M$. 
Suppose by way of contradiction that such an $N \ne M$ exists. First, it is clear that $|N|=|M|$ since $Z^N$ is a homogeneous polynomial of total degree $|M|=|N|$ in the variables $z_{1,1}, \ldots, z_{p,n}$. 
For $j \geq p$, consider the variables $z_{p,i}z'_{p,j}$ in $P_M$. They appear as $(z_{p,p}z'_{p,j})^{[\Pi_p(M)]_{p,j}}$ in $P$ since $S_p = \{ (p,p), \ldots, (p,q)\}$ and are necessarily due to the expansion of $\prod_{j=q}^q Z_{p,j}^{M_{p,j}} = \prod_{j=q}^q  \left( \sum_{t=1}^n z_{t,p}z'_{t,j}\right)^{M_{p,j}}$, which shows that $N_{p,j}\geq M_{p,j}$ for $j\geq p$.
Proceeding by induction over $i$ from $p$ to $1$, one concludes that for each $j\geq i$, it holds that $N_{i,j}\geq M_{i,j}$. 
Since the total degrees are equal, it finally implies that $N_{i,j}= M_{i,j}$ and therefore that $N=M$.

This establishes that $Z^M$ and $P_M$ are in one-to-one correspondence and that  there are not any nontrivial relations between the variables $Z_{i,j}$. The commutative algebra generated by the $Z_{i,j}$ is simply the polynomial ring in the $Z_{i,j}$. 
The fact that there are $\tbinom{pq+k-1}{k}$ different tables $M \in \N^{pq}$ with $|M|=k$ concludes the proof.
\end{proof}

\begin{corol}\label{corol:dim} For $n \geq \min(p,q)$, it holds that
\begin{align*}
F_{p,q,n}^{U(n), \leq d} = E_{p,q,n}^{\leq d}.
\end{align*}
\end{corol}

\begin{proof}
It holds that
\begin{align*}
\sum_{k=0}^d \tbinom{pq+k-1}{k} = \mathrm{dim} \,E_{p,q,n}^{ \leq d} \leq \mathrm{dim} \,F_{p,q,n}^{U(n), \leq d} \leq \sum_{k=0}^d \tbinom{pq+k-1}{k}
\end{align*}
where the equality results from Lemma \ref{lem:dimW}, the first inequality results from Lemma \ref{lem:Z-inv} and the second from Lemma \ref{lem:acts-freely}. 
The inclusion $E_{p,q,n}^{ \leq d} \subset F_{p,q,n}^{U(n), \leq d}$ from Lemma \ref{lem:Z-inv} then proves that both spaces are equal. 
\end{proof}

This establishes Theorem \ref{thm:charact-symm}.


\section{Proof of Theorem \ref{thm:trivial-kernel}}
\label{sec:proof21}

We first restate the theorem.
\begin{reptheorem}{thm:trivial-kernel}
For $p,q \geq 1, n\geq 1$, the kernel of the Laplacian operator $\Delta$ in $F_{p,q,n}^{U(n)}$ is trivial:
\begin{align*}
\ker (\Delta) \cap F_{p,q,n}^{U(n)} = \bigcap_{(i,j) \in [p]\times [q]} \ker (\Delta_{i,j}) \cap F_{p,q,n}^{U(n)}= \mathbbm{C}.
\end{align*}
Equivalently, if $\sum_{A,B} [A|B] z^A z'^{B} \in \ker(\Delta) \cap F_{p,q,n}^{U(n)}$, then $[A|B] = 0$, unless $A$ and $B$ are the all-zeroes matrices of respective size $n \times p$ and $n \times q$.
\end{reptheorem}

\begin{defn}[Weight vector]
Consider a matrix $A \in N^{n \times p}$, its weight vector is $\vec{m} =(m_1, \ldots, m_n) \in \N^n$ with $m_t = \sum_{i=1}^p A_{t,i}$. 
\end{defn}

\begin{lemma}
Let $\sum_{A,B} [A|B] z^A z'^{B} \in  F_{p,q,n}^{U(n)}$. If the weight vectors of $A \in \N^{n\times p}$ and $B \in \N^{n \times q}$ differ, then $[A|B]=0$. 
\end{lemma}
\begin{proof}
Let $t\in [m]$ be an index where the weight vectors of $A$ and $B$ differ: $\sum_{i=1}^p A_{t,i} \ne \sum_{j=1}^q B_{t,j}$. Applying Equation $(E^{t,t})$ yields
\begin{align*}
\sum_{i=1}^p A_{t,i} [A|B] = \sum_{j=1}^q B_{t,j} [A|B]
\end{align*}
which therefore implies that $[A|B]=0$.
\end{proof}

In order to establish Theorem \ref{thm:trivial-kernel}, it is therefore sufficient to consider coefficients $[A|B]$ with identical weight vectors for $A$ and $B$.

\begin{lemma} \label{lem:inv-perm}
If $\sum_{A,B} [A|B] z^A z'^{B} \in F_{p,q,n}^{U(n)}$, then for all $A \in \N^{n\times p}, B\in \N^{n\times q}$, and any permutation $\pi \in S_n$, it holds that
\begin{align*}
[\pi A | \pi B] = [A |B].
\end{align*}
\end{lemma}
\begin{proof}
The statement follows from the fact that $S_n \subset U(n)$. In particular, the polynomial $\sum_{A,B} [A|B] z^A z'^B$ should be invariant under the change of coordinates corresponding to the permutation $\pi$: $(z_{t,i}, z'_{t,j}) \leftarrow (z_{\pi(t) i}, z'_{\pi(t),j})$.
\end{proof}

As a consequence of this lemma, it is sufficient to consider coefficients $[A|B]$ where the weights of the rows of $A$, $B$ are non increasing, that is $m_1 \geq m_2 \geq \ldots m_n$ with $m_t = \sum_{i=1}^p a_{t,i} = \sum_{j=1}^q b_{t,j}$.

\begin{lemma}\label{lem:proofD0}
Let the polynomial $\sum_{A,B} [A|B] z^A z'^{B} \in F_{p,q,n}^{U(n)}$ be in the kernel of $\Delta$. Then, for all pairs $(i, j) \in [p] \times [q]$, the following equation holds
\begin{align*}
(F^{i,j}) \quad \sum_{t=1}^n \alpha_{t,i} \beta_{t,j} [A + e_{t,i}|B+e_{t,j}]=0.
\end{align*}
\end{lemma}

\begin{proof}
The equation results from applying $\Delta_{i,j}$ to $\sum_{A,B} [A|B] z^A z'^{B}$ and considering the monomial $z^A z'^B$.
\end{proof}

Recall that it is sufficient to consider matrices $A$, $B$ with non increasing row weights $\vec{m}= (m_1, \ldots, m_n)$, \ie $m_1 \geq \ldots \geq m_n\geq0$.
In the following, we fix the quantity $m = \sum_{i=1}^n m_i$ and aim to prove that $[A|B]=0$ for all possible matrices $A, B$ such that $\sum_{t,i} a_{t,i} = \sum_{t,j} b_{t,j}=m>0$.

The proof of Theorem \ref{thm:trivial-kernel} will use an induction over vectors $\vec{m}$ of fixed positive weight $|\vec{m}| = m>0$. 
Let us  introduce a total order on non increasing vectors of length $n$ and weight $m$ on $\mathcal{S}_{m,n}$, the sets of ordered integer vectors of weight $m$ and length $n$:
\begin{align*}
S_{m,n} = \left\{(m_1, \ldots, m_n) \in \N^m \: : \: m_1 \geq \ldots \geq m_n \: \text{and}\: \sum_{i=1}^n m_n=m \right\}.
\end{align*}
\begin{defn}[Total order on $\mathcal{S}_m^n$]
For $\vec{m}, \vec{m'} \in \mathcal{S}_m^n$, we write $\vec{m} \succ \vec{m}'$ if the support of $\vec{m}$ is strictly contained in the support of $\vec{m'}$, or if $\vec{m}$ is larger than $\vec{m}'$ in the lexicographic order when their supports coincide. Otherwise, $\vec{m} \prec \vec{m}'$.
\end{defn}

For instance, one has
\begin{align*}
(m,0,0,\ldots) \succ (m-1,1,0,0,\ldots) \succ (m-2, 2,0,0\ldots) \succ \ldots \succ (\lceil m/2 \rceil, \lfloor m/2 \rfloor, 0, \ldots, 0) \\
\succ (m-2,1,1) \succ (m-3,2,1) \succ \ldots
\end{align*}

An crucial property of this total order is the following. 
\begin{lemma}\label{lem:prop-order}
For $\vec{m} \ne 0$ with non increasing coefficients, define $t\in [n]$ to be the largest index such that $m_t >0$.
For all $u \in \{1, t-1\}$, one has $\vec{m} +e_u -e_t \succ \vec{m}$, \ie
\begin{align*}
(m_1, \ldots, m_u+1, \ldots m_t-1,0, \ldots) \succ (m_1, \ldots, m_u, \ldots, m_t, 0\ldots),
\end{align*}
where $e_u = (0, \ldots, 0,1, 0, \ldots, 0)$ is the vector of length $n$ with a 1 in position $u$ and 0 elsewhere.
\end{lemma}

\begin{proof}
Assume first that $m_t =1$. In that case, the vector $\vec{m} +e_u -e_t$ has a support of size $t-1$, which implies that $\vec{m} +e_u -e_t \succ \vec{m}$.
Otherwise, $m_t \geq 2$, which implies that $\vec{m} +e_u -e_t$ and  $\vec{m}$ have the same support, and the lexicographic order then implies that $\vec{m} +e_u -e_t \succ \vec{m}$.
\end{proof}

Our hypothesis induction for proving Theorem \ref{thm:trivial-kernel} is 
\begin{align*}
(H_{\vec{m}}) \quad \forall \vec{m}' \succeq \vec{m}, \, \left( A\cdot e = B\cdot e= \vec{m}' \implies [A|B]=0\right)
\end{align*}
for $\vec{m} \in \mathcal{S}_m^n$ with $m >0$.

We first prove the base case, which corresponds to $\vec{m} = (m,0, \ldots, 0)$. 
\begin{lemma}
Let $\sum_{A,B} [A|B] z^A z'^{B} \in \ker (\Delta) \cap F_{p,q,n}^{U(n)}$. If $A$ and $B$ have row weights $(m,0, \ldots, 0)$ with $m>0$, then $[A|B]=0$, \ie
\begin{align*}
\left[ \begin{smallmatrix} * &* & \cdots & *\\ 0 & 0 & \cdots & 0  \\ \vdots & \vdots & \cdots & \vdots\\ 0 & 0 & \cdots & 0  \end{smallmatrix} \right| \left. \begin{smallmatrix} * &* & \cdots & * \\ 0 & 0 & \cdots & 0 \\ \vdots & \vdots & \cdots & \vdots \\ 0 & 0 & \cdots & 0  \end{smallmatrix} \right] = 0
\end{align*}  
\end{lemma}

\begin{proof}
Consider $i \in [p], j \in [q]$ such that $a_{1,i}, b_{1,j} \geq 1$. Such indices exist because of the assumption that $m>0$. 
Let us denote $X = [A|B]$ and $X_{t,i,j} =[A - e_{1,i}+ e_{t,i}|B-e_{1,j}+e_{t,j}]$ for $t >1$, where we recall that $A$ and $B$ are zero everywhere except on their first row.  
Since $\sum_{A,B} [A|B] z^A z'^{B} \in \ker \Delta_{i,j}$, it holds that the coefficient of $z^{A-e_{1,i}} z'^{B-e_{1,j}}$ in $\Delta_{i,j} \sum_{A,B} [A|B] z^A z'^{B}$ should be zero:
\begin{align*} 
a_{1,i}b_{1,j} X + \sum_{t=2}^n X_{t,i,j}= 0. 
\end{align*}
Applying $(E^{t,t+1})$ to $[A |ÊB-e_{1,j} + e_{t,j}]$ for $t \in \{2, \ldots, n-1\}$ gives
\begin{align*}
X_{t,i,j} = X_{t+1, i,j}.
\end{align*}
Combining both equations, we obtain: 
\begin{align} 
a_{1,i}b_{1,j} X + (n-1)X_{2,i,j}= 0. \label{1114}
\end{align}
Applying $(E^{1,2})$ to $[A |ÊB-e_{1,j} + e_{2,j}]$ gives:
\begin{align}
\sum_{i} X_{2,i,j} = X\label{1114bis}
\end{align}
where by convention $[A|B]=0$ if at least one coefficient of $A$ or $B$ is negative. 
Summing Eq.~\eqref{1114} for $i\in [p]$ and inserting \eqref{1114bis} yields:
\begin{align*}
X \sum_{i=1}^p a_{1,i}b_{1,j} + (n-1) X = 0
\end{align*}
which lets us conclude that $X =0$ since $\sum_{i=1}^p a_{1,i}b_{1,j} + (n-1)  \geq n >0$.
\end{proof}

We now prove the induction step. 
\begin{lemma}
Let $\vec{m} \in \mathcal{S}_m^n$ be a vector with $m>0$.
If $(H_{\vec{m}})$ holds, then so does $(H_{\vec{m}'})$ for $\vec{m}'$ the immediate successor of $\vec{m}$ (assuming that such a successor exists).
\end{lemma}

\begin{proof}
Consider any two matrices $A, B$ with weight vector $\vec{m}$ and let us assume that $(H_{\vec{m'}})$ holds for all predecessors $\vec{m}'$ of $\vec{m}$. We wish to show that $[A|B]=0$. 
Let us denote by $t$ the largest index for which $m_t>0$. Furthermore, let us define $\mathcal{I} \subset [p], \mathcal{J} \subset [q]$ such that $a_{t,i} \ne 0$ if and only if $i \in \mathcal{I}$ and $b_{t,j} \ne 0$ if and only if $j \in \mathcal{J}$.

Let us apply  the Laplacian operator $\Delta_{i,j}$ for $i \in \mathcal{I}, j \in \mathcal{J}$. The coefficient of $z^{A - e_{t,i}}z'^{B- e_{t,j}}$ should be zero, which implies:
\begin{align}
\sum_{u=1}^{t-1} (a_{u,i}+1)(b_{u,j} +1) [A +e_{u,i} - e_{t,i}|B+e_{u,j}- e_{t,j}] + a_{t,i} b_{t,j}[A|B] \nonumber \\
+ \sum_{u=t+1}^n  (a_{u,i}+1)(b_{u,j} +1) [A - e_{t,i}+e_{u,i} |B- e_{t,j}+e_{u,j}] = 0. \label{eqn:1124}
\end{align} 
Lemma \ref{lem:prop-order} shows that the first sum is null since $ [A +e_{u,i} - e_{t,i}|B+e_{u,j}- e_{t,j}]=0$ for $u \leq t$ by the induction hypothesis. 
If $t=n$, then this already implies that $[A|B]=0$. Otherwise, assume that $t <n$. 
Lemma \ref{lem:inv-perm} shows that all the terms of the last sum are equal: $ (a_{u,i}+1)(b_{u,j} +1) [A - e_{t,i}+e_{u,i} |B- e_{t,j}+e_{u,j}] =  [A - e_{t,i}+e_{t+1,i} |B- e_{t,j}+e_{t+1,j}]$ for $u \geq t+1$. 
Eq.~\eqref{eqn:1124} then reads:
\begin{align*}
 a_{t,i} b_{t,j}[A|B] + (n-t) X_{i,j} =0 
\end{align*}
where we defined $X_{i,j} :=  [A - e_{t,i}+e_{t+1,i} |B- e_{t,j}+e_{t+1,j}]$.
Summing this equation over $i \in [p], j \in [q]$ yields:
\begin{align}
m^2 [A|B] + (n-t) \sum_{i=1}^p\sum_{j=1}^q X_{i,j} = 0.\label{eqn:1125}
\end{align}
Applying $(E^{t,t+1})$ to matrices $(A , B- e_{t,j}+e_{t+1,j})$, we obtain:
\begin{align*}
\sum_{i=1}^p  X_{i,j} = \sum_{j=1}^q b_{t,j} [A|B]. 
\end{align*}
Summing this equation over $j\in [q]$ yields:
\begin{align*}
\sum_{i=1}^p \sum_{j=1}^q  X_{i,j} = \sum_{j=1}^q m [A|B] = m^2[A|B],
\end{align*}
which, combined with Eq.~\eqref{eqn:1125}, finally shows that $[A|B] = 0$.
\end{proof}

We have therefore established Theorem \ref{thm:trivial-kernel}.


\section{Technical lemmas for the $p=q=1$ case}
\label{sec:lemmas}

In the following, we write $\partial_z$ instead of $\frac{\partial}{\partial z}$.
\begin{lemma}
\label{lemma-K+K-}
\begin{align*}
[K_-,K_+]= \hat{n}_A +\hat{n}_B + n
\end{align*}
\end{lemma}

\begin{proof}
\begin{align*}
K_- K_+ &  = \sum_{i=1}^n \partial_{z_i} \partial_{z'_i}  \sum_{j=1}^n z_j z'_j\\
&= \sum_{i=1}^n \sum_{j=1}^n  \partial_{z_i}  z_j  \partial_{z'_i} z'_j\\
&= \sum_{i=1}^n \sum_{j=1}^n (\delta_{i,j} + z_j  \partial_{z_i} )(\delta_{i,j} +z'_j    \partial_{z'_i} )\\
&= \sum_{i=1}^n  (1 + z_i \partial_{z_i} + z'_i \partial_{z'_i}  +\sum_{j=1}^n z_j  \partial_{z_i} z'_j    \partial_{z'_i} )\\
&= n + \hat{n}_A + \hat{n}_B + K_+K_-
\end{align*}

\end{proof}

\begin{lemma}
\label{lemma:nA-K+-}
\begin{align*}
[\hat{n}_A, K_\pm] = \pm K_\pm \quad \text{and} \quad [\hat{n}_B, K_\pm] = \pm K_\pm
\end{align*}
\end{lemma}

\begin{proof}
We only prove the statement for $\hat{n}_A$ since the case of $\hat{n}_B$ is identical.
\begin{align*}
K_- \hat{n}_A &= \sum_{i=1}^n  \sum_{j=1}^n \partial_{z_i} \partial_{z'_i}z_j \partial_{z_j} \\
&= \sum_{i=1}^n   \partial_{z'_i} \sum_{j=1}^n \partial_{z_i}z_j \partial_{z_j} \\
&= \sum_{i=1}^n   \partial_{z'_i} \sum_{j=1}^n \left( \delta_{i,j}\partial_{z_j} + z_j \partial_{z_i} \partial_{z_j}\right) \\
&= \sum_{i=1}^n  \sum_{j=1}^n \left( \delta_{i,j}\partial_{z_i}   \partial_{z'_i} + z_j \partial_{z_j} \partial_{z_i}   \partial_{z'_i}\right) \\
&=K_- + \hat{n}_A K_- \\
&= (\hat{n}_A+1)K_-
\end{align*}
which gives
\begin{align*}
\hat{n}_A K_- =K_- (\hat{n}_A-1).
\end{align*}
Similarly, we compute $\hat{n}_A K_+$:
\begin{align*}
\hat{n}_A K_+ &=  \sum_{i=1}^n\sum_{j=1}^n z_i \partial_{z_i} z'_j z_j \\
&=  \sum_{i=1}^n z_i  \sum_{j=1}^n z'_j \partial_{z_i}  z_j \\
&=  \sum_{i=1}^n z_i  \sum_{j=1}^n z'_j \left( \delta_{i,j} + z_j \partial_{z_i}  \right)\\
&=  \sum_{i=1}^n z_i z'_i +   \sum_{i=1}^n z_i  \sum_{j=1}^n z'_j   z_j \partial_{z_i}  \\
&= K_+ + K_+ \hat{n}_A\\
&= K_+(\hat{n}_A+1) 
\end{align*}
which gives
\begin{align*}
K_+ \hat{n}_A= (\hat{n}_A-1)K_+.
\end{align*}

\end{proof}

\begin{lemma}
\label{lemma:casimir}
\begin{align*}
[\hat{C}_2, K_0] &= [\hat{C}_2, K_+]  = [\hat{C}_2, K_-]   = 0
\end{align*}
\end{lemma}
\begin{proof}
It is sufficient to show that $[K_+K_- +K_- K_+, K_0]=0$ to prove the first claim. One has 
\begin{align*}
(K_+K_- +K_- K_+)K_0 &= K_+ K_0 K_- +K_- K_+ +K_- K_0 K_+ - K_+ K_- =  K_+ K_0 K_-  +K_- K_0 K_+ \\
&=   K_0 K_+ K_- + K_- K_+ + K_0 K_- K_+  - K_- K_+ = K_0 (K_+K_- +K_- K_+)
\end{align*}
Similarly, one can compute
\begin{align*}
\hat{C}_2 K_+ &= K_0^2 K_+ - \frac{1}{2}(K_+ K_- K_+ + K_- K_+ K_+)\\
&= K_0 K_+ K_0 + K_0 K_+ - \frac{1}{2}(K_+ K_+ K_- +2 K_+ K_0 + K_- K_+ K_+)\\
&= K_+ K_0^2 + K_+ K_0 + K_0 K_+  - \frac{1}{2}(K_+ K_+ K_- + K_+ K_- K_+ 2 K_0 K_+ + 2K_+ K_0)  \\
&= K_+ \hat{C}_2
\end{align*}
and
\begin{align*}
\hat{C}_2 K_- &= K_0^2 K_- - \frac{1}{2}(K_+ K_- K_- + K_- K_+ K_-)\\
&= K_0 K_- K_0 - K_0 K_- - \frac{1}{2}(K_+ K_- K_-  + K_- K_+ K_- -2 K_- K_0)\\
&= K_- K_0^2 - K_- K_0 - K_0 K_-  - \frac{1}{2}(K_- K_+ K_-  -2 K_0 K_- + K_- K_+K_- - 2 K_- K_0 )  \\
&= K_- \hat{C}_2
\end{align*}
\end{proof}


\begin{bibdiv}
\begin{biblist}

\bib{AA11}{inproceedings}{
      author={Aaronson, Scott},
      author={Arkhipov, Alex},
       title={The computational complexity of linear optics},
organization={ACM},
        date={2011},
   booktitle={{Proceedings of the forty-third annual ACM symposium on Theory of
  computing}},
       pages={333\ndash 342},
}

\bib{bar47}{article}{
      author={Bargmann, Valentine},
       title={{Irreducible unitary representations of the Lorentz group}},
        date={1947},
     journal={Annals of Mathematics},
       pages={568\ndash 640},
}

\bib{ber75}{article}{
      author={Berezin, Felix~A.},
       title={Quantization in complex symmetric spaces},
        date={1975},
     journal={Izvestiya: Mathematics},
      volume={9},
      number={2},
       pages={341\ndash 379},
}

\bib{car35}{inproceedings}{
      author={Cartan, Elie},
       title={Sur les domaines born{\'e}s homog{\`e}nes de l'espace de $n$
  variables complexes},
organization={Springer},
        date={1935},
   booktitle={{Abhandlungen aus dem mathematischen Seminar der Universit{\"a}t
  Hamburg}},
      volume={11},
       pages={116\ndash 162},
}

\bib{CFS02}{article}{
      author={Caves, Carlton~M.},
      author={Fuchs, Christopher~A.},
      author={Schack, R{\"u}diger},
       title={{Unknown quantum states: the quantum de Finetti representation}},
        date={2002},
     journal={Journal of Mathematical Physics},
      volume={43},
      number={9},
       pages={4537\ndash 4559},
}

\bib{chi10}{inproceedings}{
      author={Chiribella, Giulio},
       title={{On quantum estimation, quantum cloning and finite quantum de
  Finetti theorems}},
organization={Springer},
        date={2010},
   booktitle={Conference on quantum computation, communication, and
  cryptography},
       pages={9\ndash 25},
}

\bib{CKMR07}{article}{
      author={Christandl, Matthias},
      author={K{\"o}nig, Robert},
      author={Mitchison, Graeme},
      author={Renner, Renato},
       title={{One-and-a-half quantum de Finetti theorems}},
        date={2007},
     journal={Communications in Mathematical Physics},
      volume={273},
      number={2},
       pages={473\ndash 498},
}

\bib{CKR09}{article}{
      author={Christandl, Matthias},
      author={K\"{o}nig, Robert},
      author={Renner, Renato},
       title={Postselection technique for quantum channels with applications to
  quantum cryptography},
        date={2009},
     journal={Physical Review Letters},
      volume={102},
      number={2},
       pages={020504},
}

\bib{DOS07}{article}{
      author={D'Cruz, Christian},
      author={Osborne, Tobias~J.},
      author={Schack, R\"udiger},
       title={{Finite de Finetti theorem for infinite-dimensional systems}},
        date={2007},
     journal={Physical Review Letters},
      volume={98},
       pages={160406},
         url={http://link.aps.org/doi/10.1103/PhysRevLett.98.160406},
}

\bib{DF87}{inproceedings}{
      author={Diaconis, Persi},
      author={Freedman, David},
       title={{A dozen de Finetti-style results in search of a theory}},
organization={Elsevier},
        date={1987},
   booktitle={{Annales de l'IHP Probabilit{\'e}s et statistiques}},
      volume={23},
       pages={397\ndash 423},
}

\bib{FOP05}{book}{
      author={Ferraro, Alessandro},
      author={Olivares, Stefano},
      author={Paris, Matteo~G.A.},
       title={Gaussian states in quantum information},
   publisher={Bibliopolis},
        date={2005},
}

\bib{gaa73}{book}{
      author={Gaal, Steven~A.},
       title={Linear analysis and representation theory},
   publisher={Springer},
        date={1973},
}

\bib{gaz09}{book}{
      author={Gazeau, Jean-Pierre},
       title={Coherent states in quantum physics},
   publisher={Wiley},
        date={2009},
}

\bib{GN46}{article}{
      author={Gelfand, Israel~M.},
      author={Naimark, Mark},
       title={{Unitary representations of the Lorentz group}},
        date={1946},
     journal={Acad. Sci. USSR. J. Phys},
      volume={10},
       pages={93\ndash 94},
}

\bib{GW00}{book}{
      author={Goodman, Roe},
      author={Wallach, Nolan~R.},
       title={Representations and invariants of the classical groups},
   publisher={Cambridge University Press},
        date={2000},
      volume={68},
}

\bib{GG02}{article}{
      author={Grosshans, Fr\'ed\'eric},
      author={Grangier, Philippe},
       title={Continuous variable quantum cryptography using coherent states},
        date={2002},
     journal={Physical Review Letters},
      volume={88},
      number={5},
       pages={057902},
}

\bib{har47}{article}{
      author={Harish-Chandra},
       title={{Infinite irreducible representations of the Lorentz group}},
        date={1947},
     journal={Proceedings of the Royal Society of London. Series A,
  Mathematical and Physical Sciences},
       pages={372\ndash 401},
}

\bib{har51}{article}{
      author={Harish-Chandra},
       title={{Representations of Semisimple Lie Groups: IV}},
        date={1951},
     journal={Proceedings of the National Academy of Sciences},
      volume={37},
      number={10},
       pages={691},
}

\bib{har56}{article}{
      author={Harish-Chandra},
       title={{Representations of Semisimple Lie Groups VI: Integrable and
  Square-Integrable Representations}},
        date={1956},
     journal={American Journal of Mathematics},
       pages={564\ndash 628},
}

\bib{har13}{article}{
      author={Harrow, Aram~W.},
       title={The church of the symmetric subspace},
        date={2013},
     journal={arXiv preprint 1308.6595},
}

\bib{hel79}{book}{
      author={Helgason, Sigurdur},
       title={{Differential geometry, Lie groups, and symmetric spaces}},
   publisher={Academic press},
        date={1979},
      volume={80},
}

\bib{JM17}{article}{
      author={Jiang, Tiefeng},
      author={Ma, Yutao},
       title={Distances between random orthogonal matrices and independent
  normals},
        date={2017},
     journal={arXiv preprint arXiv:1704.05205},
}

\bib{KM09}{article}{
      author={K\"onig, Robert},
      author={Mitchison, Graeme},
       title={{A most compendious and facile quantum de Finetti theorem}},
        date={2009},
     journal={{Journal of Mathematical Physics}},
      volume={50},
      number={1},
       pages={012105},
}

\bib{KR05}{article}{
      author={K{\"o}nig, Robert},
      author={Renner, Renato},
       title={{A de Finetti representation for finite symmetric quantum
  states}},
        date={2005},
     journal={Journal of Mathematical Physics},
      volume={46},
      number={12},
       pages={122108},
}

\bib{lev17}{article}{
      author={Leverrier, Anthony},
       title={{Security of Continuous-Variable Quantum Key Distribution via a
  Gaussian de Finetti Reduction}},
     journal={Physical Review Letters},
      volume={118},
       pages={200501},
}

\bib{LGRC13}{article}{
      author={Leverrier, Anthony},
      author={Garc\'ia-Patr\'on, Ra\'ul},
      author={Renner, Renato},
      author={Cerf, Nicolas~J.},
       title={Security of continuous-variable quantum key distribution against
  general attacks},
        date={2013},
     journal={Physical Review Letters},
      volume={110},
       pages={030502},
}

\bib{per72}{article}{
      author={Perelomov, Askold~M.},
       title={{Coherent states for arbitrary Lie group}},
        date={1972},
     journal={Communications in Mathematical Physics},
      volume={26},
      number={3},
       pages={222\ndash 236},
}

\bib{per86}{book}{
      author={Perelomov, Askold~M.},
       title={Generalized coherent states and their applications},
   publisher={Springer},
        date={1986},
}

\bib{pur94}{article}{
      author={Puri, Ravinder~R.},
       title={{$SU(m,n)$ coherent states in the bosonic representation and
  their generation in optical parametric processes}},
        date={1994},
     journal={Physical Review A},
      volume={50},
       pages={5309\ndash 5316},
         url={http://link.aps.org/doi/10.1103/PhysRevA.50.5309},
}

\bib{ren08}{article}{
      author={Renner, Renato},
       title={Security of quantum key distribution},
        date={2008},
     journal={International Journal of Quantum Information},
      volume={6},
      number={01},
       pages={1\ndash 127},
}

\bib{RC09}{article}{
      author={Renner, Renato},
      author={Cirac, Juan~Ignacio},
       title={{de Finetti Representation Theorem for Infinite-Dimensional
  Quantum Systems and Applications to Quantum Cryptography}},
        date={2009},
     journal={Phys. Rev. Lett.},
      volume={102},
      number={11},
       pages={110504},
}

\bib{ren07}{article}{
      author={Renner, Renner},
       title={Symmetry of large physical systems implies independence of
  subsystems},
        date={2007},
     journal={Nat. Phys.},
      volume={3},
      number={9},
       pages={645\ndash 649},
}

\bib{SBC08}{article}{
      author={Scarani, Valerio},
      author={Bechmann-Pasquinucci, Helle},
      author={Cerf, Nicolas~J.},
      author={Du{\v{s}}ek, Miloslav},
      author={L{\"u}tkenhaus, Norbert},
      author={Peev, Momtchil},
       title={The security of practical quantum key distribution},
        date={2009},
     journal={Reviews of Modern Physics},
      volume={81},
      number={3},
       pages={1301},
}

\bib{wat16}{book}{
      author={Watrous, John},
       title={Theory of quantum information},
        date={2016},
}

\bib{WLB04}{article}{
      author={Weedbrook, Christian},
      author={Lance, Andrew~M.},
      author={Bowen, Warwick~P.},
      author={Symul, Thomas},
      author={Ralph, Timothy~C.},
      author={Lam, Ping~Koy},
       title={Quantum cryptography without switching},
        date={2004},
     journal={Physical Review Letters},
      volume={93},
      number={17},
       pages={170504},
}

\bib{WPG12}{article}{
      author={Weedbrook, Christian},
      author={Pirandola, Stefano},
      author={Garc{\'i}a-Patr\'on, Ra\'ul},
      author={Cerf, Nicolas~J.},
      author={Ralph, Timothy~C.},
      author={Shapiro, Jeffrey~H.},
      author={Lloyd, Seth},
       title={Gaussian quantum information},
        date={2012},
     journal={Reviews of Modern Physics},
      volume={84},
       pages={621\ndash 669},
}

\end{biblist}
\end{bibdiv}


\end{document}